\documentclass[10pt,journal,compsoc]{IEEEtran}

\usepackage{cite} 
\usepackage{graphicx} 
\usepackage{subfigure} 
\usepackage{caption2} 
\usepackage{amsmath}
\usepackage{amsmath}
\usepackage{amssymb}
\usepackage{amsfonts}
\usepackage{graphicx}
\usepackage{epsfig}
\usepackage{cite}
\usepackage{multirow}
\usepackage{tabularx,array}
\usepackage{longtable}
\usepackage{xcolor}
\usepackage{bm}
\usepackage{booktabs}  
\usepackage[ruled,vlined,lined,ruled,linesnumbered]{algorithm2e}
\SetKwRepeat{Do}{do}{while}

\newtheorem{theorem}{Theorem}

\newtheorem{proof}{Proof}

\definecolor{LightBlue}{RGB}{0,111,191}

\begin{document}
	\title{Attention-based SIC Ordering and Power Allocation for Non-orthogonal Multiple Access Networks}
	
	\author{Liang~Huang,~\IEEEmembership{Senior Member,~IEEE,}
		Bincheng~Zhu,
		Runkai~Nan,
		Kaikai~Chi,~\IEEEmembership{Senior Member,~IEEE,}
		Yuan~Wu,~\IEEEmembership{Senior Member,~IEEE}
		\IEEEcompsocitemizethanks{\IEEEcompsocthanksitem L. Huang, B. Zhu, R. Nan, and K. Chi are with the School of Computer Science and Technology, Zhejiang University of Technology, China (e-mail: \{lianghuang, bczhu, rknan, kkchi\}@zjut.edu.cn).\protect\\
			\IEEEcompsocthanksitem Y. Wu is with the State Key Laboratory of Internet of Things for Smart City, University of Macau, Macau, China, and also with the Department of Computer and Information Science, University of Macau, Macau, China (e-mail: yuanwu@um.edu.mo).
		}
	}

	\IEEEtitleabstractindextext{
		\begin{abstract}
			Non-orthogonal multiple access (NOMA) emerges as a superior technology for enhancing spectral efficiency, reducing latency, and improving connectivity compared to orthogonal multiple access. In NOMA networks, successive interference cancellation (SIC) plays a crucial role in decoding user signals sequentially. The challenge lies in the joint optimization of SIC ordering and power allocation, a task made complex by the factorial nature of ordering combinations. {This study introduces an innovative solution, the Attention-based SIC Ordering and Power Allocation (ASOPA) framework, targeting an uplink NOMA network with dynamic SIC ordering. ASOPA aims to maximize weighted proportional fairness by employing deep reinforcement learning, strategically decomposing the problem into two manageable subproblems: SIC ordering optimization and optimal power allocation. Our approach utilizes an attention-based neural network, which processes instantaneous channel gains and user weights to determine the SIC decoding sequence for each user. A baseline network, serving as a mimic model, aids in the reinforcement learning process. Once the SIC ordering is established, the power allocation subproblem transforms into a convex optimization problem, enabling efficient calculation of optimal transmit power for all users. 
			Extensive simulations validate ASOPA's efficacy, demonstrating a performance closely paralleling the exhaustive method, with over 97\% confidence in normalized network utility. 
			{
			Compared to the current state-of-the-art implementation, i.e., Tabu search, ASOPA achieves over 97.5\% network utility of Tabu search. Furthermore, ASOPA is two orders magnitude less execution latency than Tabu search when $N=10$ and even three orders magnitude less execution latency less than Tabu search when $N=20$.}
			Notably, ASOPA maintains a low execution latency of approximately 50 milliseconds in a ten-user NOMA network, aligning with static SIC ordering algorithms. 
			Furthermore, ASOPA demonstrates superior performance over baseline algorithms besides Tabu search in various NOMA network configurations, including scenarios with imperfect channel state information, multiple base stations, and multiple-antenna setups. 
			Such results underscore ASOPA's robustness and effectiveness, highlighting its ability to excel across various NOMA network environments. The complete source code for ASOPA is accessible at https://github.com/Jil-Menzerna/ASOPA.}
		\end{abstract}
		}
	
	\maketitle
	
	\IEEEdisplaynontitleabstractindextext
	
	\IEEEpeerreviewmaketitle
	
	\section{Introduction}
	\IEEEPARstart{W}{ith} the rapid development of online gaming, augmented and virtual reality, 3-dimensional media, and Internet of Things, wireless network traffic has increased significantly, which is a challenge for orthogonal multiple access (OMA) schemes. 
	{
	To meet the growing demand, the next generation of wireless networks exploits advanced multiple access technologies \cite{yang2021MWC, Liu2022,zakeri2021robust}, including the non-orthogonal multiple access (NOMA) \cite{LinZ2019NOMA,DM2020} and the rate-splitting multiple access (RSMA) \cite{LinZ2021RSMA}. }
	Combining NOMA with other technologies, such as cognitive radios, unmanned aerial vehicles, mobile edge computing, simultaneous wireless information and power transfer, etc., can bring considerable advantages \cite{akbar2021noma,WuY2022}. Specifically, NOMA gives improved spectral efficiency, energy efficiency, higher data rates, massive connectivity, and diversity of wireless service \cite{maraqa2020survey, vaezi2019non}. 
	
	For the uplink power-domain NOMA network, users simultaneously transmit their data over the same frequency resource, so there is inter-user interference due to the broadcast and superposition nature of the wireless medium \cite{islam2016power,ZhaiD2018}. 	
	At the base station (BS) of uplink NOMA-based networks, the successive interference cancellation (SIC) technique decodes different users' messages from the received signal. Using SIC, the BS sequentially decodes users' messages according to a particular order. When decoding a user's message, the remaining undecoded signal is treated as interference. After decoding, a user's message will be subtracted from the received signal.
	The procedure continues until all users' messages are decoded according to a specific SIC order.
	Although different SIC orderings generate the same sum throughput of the NOMA network with a single BS \cite{chi2019energy}, they affect the throughput of each user \cite{huang2022tvt, you2020note,zakeri2019joint,cui2018application}. The earlier a user's signal is decoded, the more substantial interference it suffers. 
	When the quality of service of an individual user matters, it is necessary to optimize the SIC order for better performance metrics, e.g., outage probability \cite{gao2017theoretical,jiang2018sic,gao2018analysis}, latency \cite{qian2019optimal}, and energy consumption \cite{ding2020unveiling2, ZhangL2022}.
	
	For a NOMA network, the joint optimization of SIC ordering and resource allocation is a Non-deterministic Polynomial (NP) hard problem \cite{qian2019optimal}. 
	Many researchers decomposed the joint optimization into SIC ordering optimization and resource allocation.
	For SIC ordering, the total number of decoding orderings is the factorial of the number of users. The exhaustive method obtains the optimal SIC ordering by enumerating all possible SIC orderings and is limited to small-scale scenarios, i.e., with five users\cite{hu2019joint}.
	To tradeoff performance and computational complexity, there are two types of heuristic methods for SIC ordering in many works, i.e., static heuristic methods and dynamic heuristic methods. 
	Some researchers \cite{duan2019resource,pan2017resource,sun2018robust,gao2017theoretical,you2020note,gao2018analysis,chen2014evaluations, Lai2022, Cui2023,yakou2015downlink,jiang2018sic,wang2021tcomm} adopted a static SIC ordering order with respect to a single metric and optimized the resource allocation when considering specific NOMA-based wireless networks. The execution latency of static SIC ordering methods is very low, which is close to the conventional SIC ordering algorithms, i.e., the descending order and ascending order of channel quality.
	However, these static SIC ordering methods can not cope with complex NOMA wireless scenarios, as finding the corresponding single metric is problematic.
	Some other works \cite{qian2019optimal,LiuZ2024,QianL2024,zakeri2019joint,zakeri2021robust} tried to iteratively search for the SIC ordering, e.g., greedy insertion and linear relaxation.
	While the dynamic heuristic methods apply to most NOMA wireless scenarios and achieve better performance compared to static SIC ordering algorithms, they still have high complexity in the joint optimization problems due to repeatedly optimizing the resource optimization.
	
	Recently, deep learning has emerged as a promising approach for making near-optimal decisions. With a large labelled dataset, it can achieve substantial performance through supervised learning. However, in the joint optimization of SIC ordering and resource allocation, obtaining the optimal solution labels is challenging, particularly in large-scale NOMA wireless networks.
	In this regard, deep reinforcement learning (DRL) is a suitable technique that can train a model using a dataset without labels. Some recent studies have utilized DRL techniques to efficiently solve computation-intensive resource allocation problems in wireless networks, such as Deep Q-Network \cite{huang2019dcn,wu2021hybrid}, actor-critic algorithm \cite{zhou2022iotj,tuli2020dynamic}, deep deterministic policy gradient \cite{ZouY2022, ZhaoT2022}, and proximal policy optimization \cite{samir2021optimizing,li2020trajectory}.
	{Recently, a pioneering deep reinforcement learning-based framework named DROO is introduced to address the hybrid integer-continuous challenge in mobile edge computing \cite{huang2019deep}. DROO ingeniously splits the primary optimization issue into two subproblems: a zero–one binary offloading decision and a continuous resource allocation task. These subproblems are then individually managed through a model-free learning module and a model-based optimization module, respectively \cite{Zhang2021}. Nevertheless, DROO and its subsequent iterations \cite{BiS2021,Li2022}, are constrained by their reliance on quantization modules that exclusively produce binary decisions. This limitation becomes particularly evident in their inability to permute SIC ordering for NOMA networks. The permutation complexity for SIC orderings is factorial, presenting a significantly more challenging scenario than binary decisions. This complexity forms the basis of our motivation to develop a solution that adeptly handles the joint SIC ordering and power allocation problem, aiming to achieve near-optimal performance. This is particularly crucial in meeting the real-time requirements of NOMA networks, where efficiently managing the factorial complexity is essential for optimal system operation.}
	
	In this paper, we consider the uplink NOMA with different weighted users. Since transmit power of the NOMA wireless network is typically shared in a best-effort fashion, we aim to ensure fairness across multiple users. We optimize the SIC ordering and users' transmit power to maximize the weighted proportional fairness function. To tackle this challenging problem, a novel Attention-based SIC ordering and power allocation (ASOPA) framework is proposed, which leverages both DRL and optimization theory. { To assess the effectiveness of ASOPA, we conduct comparative analyses against a range of baseline algorithms. These comparisons focus on two key metrics: the network utility achieved and the execution duration of ASOPA. Furthermore, to demonstrate the wide-ranging applicability and versatility of ASOPA, the ASOPA is applied with various NOMA network configurations. This extension effectively highlights ASOPA's adaptability across diverse network structures within the NOMA framework, underlining its robustness and practical utility in diverse network conditions.
	}
	
	Our main contributions can be summarized as follows:	
	\begin{itemize}

		\item
		We formulate a joint optimization problem of SIC ordering and power allocation to maximize the weighted proportional fairness on an uplink NOMA wireless network. To solve this problem efficiently, it is decomposed into two subproblems: a SIC ordering subproblem and a power allocation subproblem under the given SIC ordering. This decomposition approach enables us to leverage DRL and convex optimization for optimal performance.
		
		\item 
		We propose the ASOPA framework to solve the joint optimization problem. This framework comprises three components: an attention-based actor network, a convex optimization module, and a baseline network. The actor network generates the SIC ordering, while the optimization module allocates users' transmit power. The baseline network is used to train and reinforce the actor network. The ASOPA framework is designed to generate feasible solutions that meet all physical constraints, ensuring optimal performance. 
		
		\item
		{We conduct comprehensive numerical experiments to validate the effectiveness of the ASOPA framework. The findings reveal that ASOPA attains near-optimal performance, fulfilling the real-time demands of NOMA networks. Notably, the normalized network utility achieved by ASOPA has a confidence interval exceeding 99\%, closely mirroring the performance of exhaustive methods. In terms of execution latency, ASOPA is on par with static SIC ordering algorithms, with an approximate duration of 50 ms in a ten-user NOMA network. To highlight ASOPA's broad applicability and versatility, we extend its application to include scenarios with imperfect channel state information (CSI), networks comprising multiple BSs, and systems with multiple-antenna setups. In each of these diverse environments, ASOPA consistently outperforms the baseline algorithms, showcasing its robustness and effectiveness across different NOMA network configurations.}
		
	\end{itemize}
	
	The rest of this paper is organized as follows. The related work is introduced in Section \ref{Sec:RelatedWork}. 
	Section \ref{sec:model} gives the system model and problem formulation. 
	Section \ref{sec:DRL} introduces ASOPA. 
	The numerical results are given in Section \ref{sec:NumResults}.
	Finally, Section \ref{sec:conclusion} concludes the paper.

	\section{Related Work} \label{Sec:RelatedWork}
	Most of the existing works on NOMA use fixed SIC ordering according to channel conditions. Specifically, the descending order of channel quality is usually used in uplink NOMA \cite{duan2019resource,wang2021tcomm}, and the ascending order of channel quality is generally used in downlink NOMA \cite{pan2017resource,sun2018robust}. However, in many scenarios, using the above fixed ascending or descending order for SIC might not be optimal. To achieve optimal performance, \cite{hu2019joint} used the exhaustive search to find the optimal decoding order. However, the computational complexity of the exhaustive search is at least $O(N!)$, and its usage is limited to small-scale NOMA wireless scenarios, i.e., less than five users.
	To balance performance and computational complexity, the following works further optimized the SIC ordering in a NOMA network, which affects different performance metrics, e.g., outage probability\cite{gao2017theoretical, jiang2018sic, gao2018analysis}, throughput\cite{yakou2015downlink,huang2022tvt, you2020note, zakeri2019joint, cui2018application},  latency\cite{qian2019optimal}, energy consumption\cite{ding2020unveiling2, ZhangL2022} and the data rate of a particular user\cite{hu2019joint}. After elaborately designing the SIC ordering algorithm, Qian {\em et al.} \cite{qian2019optimal} showed that by optimizing the SIC ordering, the min-max execution latency could be reduced by ten times compared to the best comparison method. Also, \cite{zakeri2019joint} showed that optimizing the SIC ordering can get a 48\% improvement in the sum rate over the fixed SIC ordering. The above work of SIC ordering optimization can be classified into two types: static SIC ordering and dynamic SIC ordering.
	
	\subsection{Static SIC Ordering}
	In addition to wireless channel quality, some works have proposed using a static decoding order based on other performance metrics specific to certain problems and scenarios. For example, the descending order of received user signal power \cite{gao2017theoretical,you2020note,gao2018analysis}, the descending order of predicted user throughput \cite{chen2014evaluations}, the decreasing order of channel gain normalized by noise and interference power \cite{Lai2022, Cui2023}, the ascending order of average channel gain from user to the base station \cite{yakou2015downlink}, and the ascending order of a user's maximum secrecy throughput \cite{jiang2018sic}.
	For the uplink NOMA scenario, \cite{jiang2018sic} adopted the ascending order of user's maximum secrecy throughput as the SIC ordering to achieve secrecy transmission in eavesdropper scenarios. \cite{gao2017theoretical} used the descending order of received user signal power as the SIC ordering and derived the closed-form expression of the outage probability for a single base station and three-user system. Building on this work, \cite{gao2018analysis} considered the single base station and multiple active user scenario in the case of imperfect channel state information, using the descending order of estimated instantaneous received signal power as the SIC ordering.
	Although these static SIC ordering methods can achieve excellent performance with small execution latency, they may suffer considerable performance degradation on complex NOMA networks due to the difficulty of finding a suitable metric for ordering.
	
	\subsection{Dynamic SIC Ordering}
	In addition to works that use a static SIC ordering based on specific problems and scenarios, a few studies have proposed iteratively searching heuristic algorithms to optimize the SIC ordering.
	For instance, \cite{zakeri2021robust} and \cite{zakeri2019joint} introduced binary variables to represent the SIC ordering, solving the problem via variable relaxation with compensation for performance degradation and as an integer linear programming problem, respectively.
	{
	\cite{LiuZ2024} adapt a permutation-based genetic algorithm to optimize the SIC ordering.
	\cite{qian2019optimal} used the greedy meta-scheduling technique to develop a low-complexity and easy-to-implement SIC ordering algorithm. This algorithm sequentially inserts users into the existing ordering, tries every possible insertion position for each user, and chooses the position that provides the most significant benefit.
	\cite{QianL2024} utilized a heuristic tabu search to optimize the SIC ordering in an iterative process. At each iteration, tabu search swapped the SIC ordering of any two users and selected the best one. To the best of our knowledge, \cite{QianL2024} is the state-of-the-art algorithm for dynamic SIC ordering, it achieves the near-optimal performance, but suffers from high computational complexity due to the large number of iterations.}
	While these dynamic SIC ordering methods apply to most NOMA wireless scenarios and achieve better performance than static SIC ordering, they all involve iterative updates and repeatedly solve resource allocation subproblems, resulting in high computational complexity. 

	\section{System model and problem formulation} \label{sec:model}
	\subsection{System Model}\label{342}
	As shown in Fig. \ref{fig:systemmodel}, we consider an uplink NOMA network with a central-located BS and $N$ active users, denoted as a set $\mathcal{N}=\{1,2,\dots, N\}$, where each user has a single antenna. Users have a stable power supply, and all $N$ users simultaneously transmit their information to the BS by NOMA.
	\begin{figure}
		\centering
		\includegraphics[width=1\linewidth]{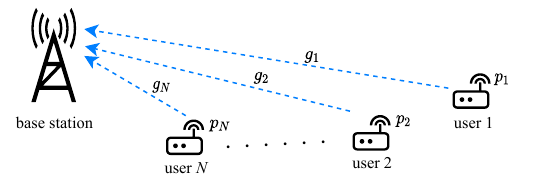}
		\caption{An uplink NOMA network with one BS and $N$ users.}
		\label{fig:systemmodel}
	\end{figure}
	
	The system time is divided into consecutive slots of equal length, smaller than the channel coherence time. We assume that the wireless channel gain is constant in each time slot and may vary across different slots. Without loss of generality, the slot length is normalized for brevity. 
	
	The BS employs SIC in a successive order to decode users' signals.
	{ In our framework, we define a function $\xi(n) = i$ and its inverse $\pi(i) = n$, which establish a mapping between a user's index $n$ and its corresponding decoding order $i$. For instance, $\xi(3) = 2$ indicates that the user~3 is decoded second in the sequence. Conversely, $\pi(2) = 3$ indicates that the second user to be decoded in the order is user~3. This bi-directional mapping functions serve to clearly delineate the relationship between the decoding sequence and the specific users in the network.} Hence, the signal-to-interference-plus-noise ratio of the user $n$ can be expressed as
	\begin{equation}
		\phi_n=\frac{p_ng_n}{\sum\limits_{\xi\left(n'\right)>\xi\left(n\right),\forall n' \in\mathcal{N}}{p_{n'}g_{n'}}+N_0},
	\end{equation}
	where $N_0$ is the power of the additive Gaussian noise at the BS, $g_n$ denotes the wireless channel gain between user $n$ and the BS at a tagged time slot, and $p_n$ denotes the transmit power of user $n$ sending its information. So the data rate of user $n$ can be expressed as 
	\begin{equation}
		R_n=B\log_2(1+\phi_n),
	\end{equation}
	where $B$ denotes the communication bandwidth.
	
	We summarize essential notations used throughout this paper in Table~\ref{table:notations}.
	\begin{table}
		\setlength{\abovecaptionskip}{0cm}
		\setlength{\belowcaptionskip}{0.1cm}
		\renewcommand{\arraystretch}{1.3}
		\centering
		\caption{Notations}\label{table:notations}
		
		\begin{tabular}{cm{6.5cm}}
			\hline		
			Notation & Definition \\
			\hline
			$N$&The number of users  \\
			\hline
			$\mathcal{N}$ &The set of users\\
			\hline
			$g_n$&The wireless channel gain between the user $n$ and the BS\\
			\hline
			$p_n$&The transmit power of  user $n$\\
			\hline
			$\mathbf{p}$&The vector representation of the power allocation\\
			\hline
			$P_n^{max}$&The maximum power that user $n$ can achieve\\
			\hline
			$\bm{\pi}$&The SIC ordering of all users\\
			\hline
			$\bm{\pi}^{BL}$&The SIC ordering generated by the baseline network\\
			\hline
			$\Pi$&The set of all possible SIC orderings\\
			\hline
			$\xi\left(n\right)$&The order of user $n$ to be decoded\\
			\hline
			$N_0$&The power of the additive Gaussian noise at the BS\\
			\hline
			$\phi_n$&The signal-to-interference-plus-noise ratio of the user $n$\\
			\hline
			$B$&The communication bandwidth\\
			\hline
			$R_n$&The data rate of user $n$\\

			\hline
			$w_n$&The weight of user $n$\\
			
			\hline
			$\mathbf{X}$&The representation of all users' information\\
			\hline
			$\theta$&The parameters of the actor network\\
			\hline
			$\theta'$& The parameters of the baseline network\\

			\hline
			$\mathbf{E}$&The embedding of all users generated by the encoder\\
			
			\hline
			$\overline{\mathbf{e}}$&The global information embedding which is the mean of $\mathbf{e_n},\forall n\in\mathcal{N}$ \\
			\hline
			$\bm{\ell}_t$&The probability of users being selected at iteration $t$\\
			\hline
			$\mathbf{q}, \mathbf{K}, \mathbf{V}$ & The query, key, and value\\
			\hline
			$d_k$&The dimension of $\mathbf{q}$ and $\mathbf{k}$\\
			\hline
			$d_e$ & The dimension of each user's embedding $\mathbf{e}_n$ \\
			\hline
			$S\left({\theta}|\mathbf{X} \right)$&The expected objective value for input $\mathbf{X}$ under the network parameters $\theta$\\
			\hline
			$z_\theta\left(\bm{\pi}|\mathbf{X}\right)$&The probability of $\bm{\pi}$ generated by the actor network $\theta$ for the input $\mathbf{X}$\\
			\hline
			$R(\bm{\pi}|\mathbf{X})$&The network utility under the given SIC ordering $\bm{\pi}$ for $\mathbf{X}$\\
			\hline
			$\left | \bm{\tau} \right | $&The batch size\\
			\hline
			$\tau$ &The index of sample in a training batch\\
			\hline
			$\bm{\tau}$& The set of training batch\\
			\hline
		\end{tabular}
	\end{table}
	
	\subsection{Problem Formulation}
	This paper aims to achieve weighted proportional fairness across multiple users. Proportional fairness can be achieved by maximizing individual rates with a logarithmic utility function \cite{Kelly1998}. In addition, considering users' different priorities, individual weights are assigned to each user to achieve weighted proportional fairness \cite{LiW2014,ChenL2019,IHhou2014}. Thus, the aim is to maximize the weighted sum of the logarithmic throughput of different users by jointly optimizing the SIC ordering $\bm{\pi}$ and the power allocation $\textbf{p}$, denoted as the network utility $R(\bm{\pi},\textbf{p})$. This optimization problem can be expressed mathematically as:
	\begin{subequations}
		\begin{align}
			\textbf{P0}:R(\bm{\pi},\textbf{p}) = \max\limits_{\mathbf{\bm{\pi}},\textbf{p}}\ &\sum\limits_{n=1}^{N}{w_n\ln R_n}\label{P0_objective}\\
			s.t.\ &0<p_n\le P_n^{max},\forall n \in\mathcal{N}\label{P0_constraint_power},\\
			&\bm{\pi}\in\Pi\label{P0_constraint_sicorder},
		\end{align}
	\end{subequations}
	where $w_n$ is the weight of user $n$. $\bm{\pi}=[\pi_1, \pi_2,\cdots,\pi_N]$ indicates the SIC order, 
	{where we denote ${\pi _i}=\pi(i)$ for brevity. }
	${\Pi}$ is the permutation set of all possible SIC orderings with size factorial $N$, represented as $N!$. $\textbf{p}=[p_1,p_2,\dots,p_N]$ is the power allocation. \eqref{P0_constraint_power} is the power constraint for each user $n$, where $P_n^{max}$ is the maximum power that user $n$ can achieve. \eqref{P0_constraint_sicorder} is the constraint for $\bm{\pi}$.

	The problem \textbf{P0} involves combinatorial optimization and continuous numerical optimization, which is NP-hard. 
	To effectively solve the problem \textbf{P0}, we decompose it into the SIC ordering optimization and the optimization of power allocation under the given SIC ordering, as shown in Fig.~\ref{fig:two-level}:
	\begin{figure}
		\centering
		\includegraphics[width=1\linewidth]{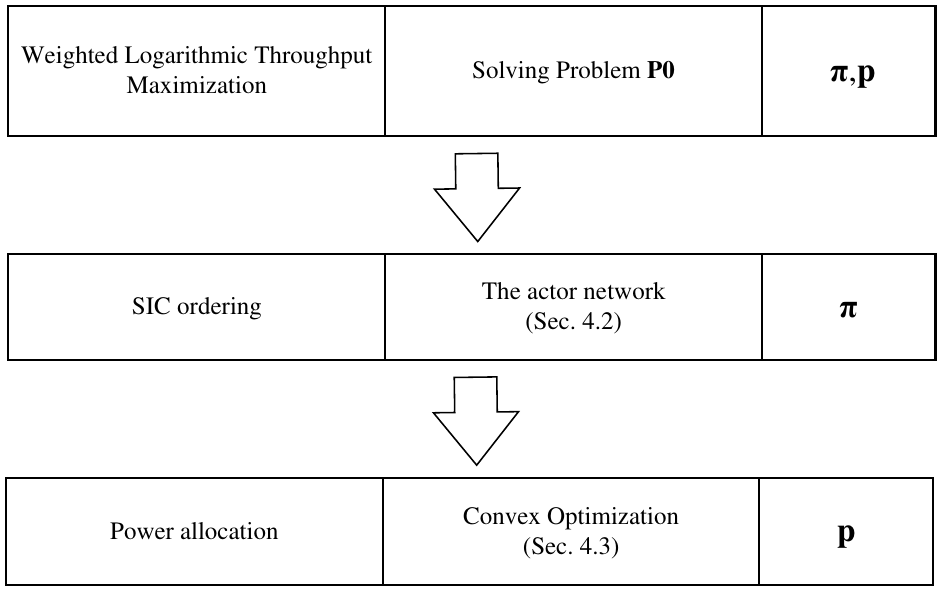}
		\caption{The two-level optimization structure of solving problem \textbf{P0}.}
		\label{fig:two-level}
	\end{figure}
	
	\begin{itemize}
		\item \textit{SIC Ordering}: 
		It is computationally expensive to iteratively search for the optimal SIC ordering from $\Pi$ at the $N!$ scale. We tell that one SIC ordering outperforms another one by solving the power allocation problems \textbf{P1} and comparing their utilities $R(\bm{\pi})$. However, classical comparison-based sorting algorithms cannot perform better than $O(n \log n)$ on average \cite{cormen2022introduction}, which requires repeatedly solving \textbf{P1}, resulting in long execution latency. In this paper, the deep reinforcement learning method is adopt to generate the SIC ordering before the power allocation.
		
		\item \textit{Power Allocation}: 	When the SIC ordering $\bm{\pi}	$ is determined, we only need to solve the power allocation $\mathbf{p}$, as follows:
		\begin{subequations}
			\begin{align}
				\textbf{P1}:R(\bm{\pi})=\max\limits_{\textbf{p}}\ &\sum\limits_{n=1}^{N}{w_n\ln{R_n}}\label{P1_objective}\\
				s.t.\ &0< p_n\le P_n^{max},\forall n \in\mathcal{N}\label{P1_constraint_power}.
			\end{align}
		\end{subequations}
		
		We can solve this power allocation sub-problem \textbf{P1} by converting it to a convex problem and using the inter-point method.
	\end{itemize}
	
	In the next Section, these two subproblems are solved by taking advantage of DRL and convex optimization.
	
	\section{Algorithm Design}  \label{sec:DRL}
	\begin{figure*}
		\centering
		\includegraphics[width=1\linewidth]{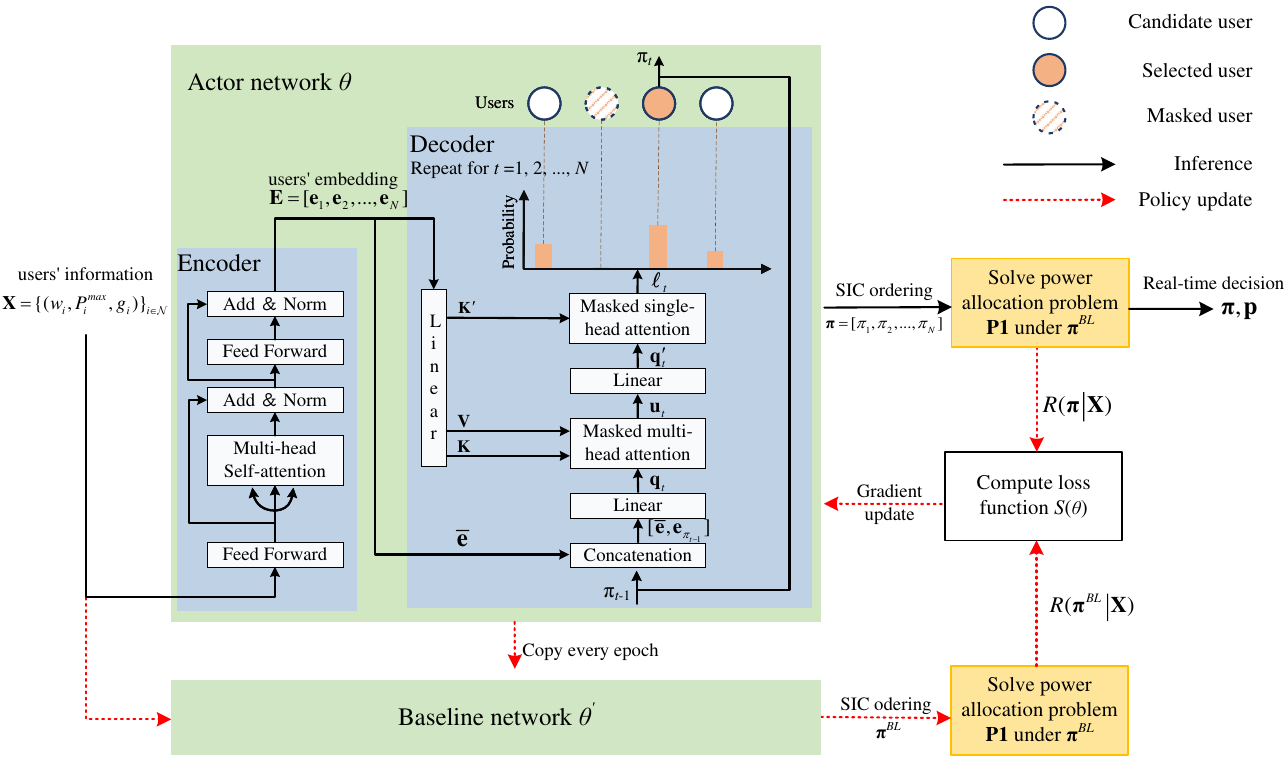}
		\caption{The schematics of the proposed ASOPA framework.}
		\label{fig:algorithm_overview}
	\end{figure*}
	\subsection{Algorithm Overview}	
	Fig.~\ref{fig:algorithm_overview} shows the schematics of the proposed ASOPA framework. It uses an encoder-decoder-based actor network $\theta$ to generate the SIC ordering sequentially and uses optimization techniques to optimize all users' transmit power. The design of the actor network follows from the pointer network \cite{vinyals2015pointer,kool2018attention} for routing problems.
	{
		In Fig.~\ref{fig:algorithm_overview}, the black solid lines depict the inference process of ASOPA, which necessitates real-time network parameters. These parameters include each user's weight $w_n$, maximum power $P_n^{max}$, and channel gain $g_n$. The complete set of users' information, denoted by $\mathbf{X}=\{ (w_i,P_i^{max}, g_i)\}_{i\in \mathcal{N}}$, is fed into the actor network to generate the SIC ordering $\bm{\pi}$.
		Following the determination of the SIC ordering $\bm{\pi}$, we solve the sub-problem \textbf{P1} for the optimal power allocation $\mathbf{p}$ through convex optimization. Then, ASOPA outputs the network decision, represented as \( (\bm{\pi}, \mathbf{p}) \), based on the instantaneous user information \( \mathbf{X} \).}

	{
		The red dotted lines in Fig.~\ref{fig:algorithm_overview} illustrate the policy update process in ASOPA. We employ a replica of the actor network, referred to as the baseline network $\theta'$, and train the actor network using the REINFORCE algorithm. Each users' information $ \bf{X}$ is fed into both the actor and baseline networks. This process yields the output SIC orderings $\bm{\pi}$ and the baseline SIC orderings $\bm{\pi}^{BL}$. Subsequently, \( R\left( {\bm{\pi } |{\bf{X}}} \right)\) and \( R\left( \bm{\pi } ^{BL}|\bf{X} \right) \) are obtained by solving problem \textbf{P1}, which are used to compute the loss function. The backpropagation method is employed to update the parameters \( \theta \) of the actor network. The procedures of ASOPA are detailed in the following subsections.
	}

	\subsection{SIC Ordering}\label{sec_sic_ordering_selection_network}
	As shown in Fig.~\ref{fig:algorithm_overview}, the SIC ordering $\bm{\pi}$ is generated by the actor network composed of an encoder and a decoder as follows:
	{\begin{enumerate}
			\item{The encoder} takes users' information $\mathbf{X}=[\mathbf{x}_1,\mathbf{x}_2,\dots,\mathbf{x}_N]$ as input and outputs users' embedding $\mathbf{E}=[\mathbf{e}_1,\mathbf{e}_2,...,\mathbf{e}_N]$ using self-attention layers. The global embedding ${\overline{\mathbf e}}$ is then calculated as the average of $\mathbf{E}$.
			\item{The decoder} generates the SIC ordering in an iterative process. At each iteration $t$, by utilizing the cross-attention layers, the decoder generates the probability of all users according to $[{\overline{\mathbf e}},{\mathbf{e}}_{\pi_{t-1}}]$ and $\mathbf{E}$. By masking the previously selected users, the probability of the remaining users being selected is calculated using the softmax function, allowing the decoder to determine the current user's index \(\pi_t\). At the end of each iteration, the decoder updates users selected for masking and takes ${\mathbf{e}}_{\pi_t}$ as input to next iteration. It iterates $N$ times to obtain the complete SIC ordering $\bm{\pi}=[\pi_1,\pi_2,...,\pi_N]$.
	\end{enumerate}
}

	
	{The integration of an attention scheme in the encoder and an iterative decoding scheme in the decoder empowers ASOPA to effectively manage a varying number of users in NOMA networks. This approach ensures adaptability and responsiveness to user dynamics, maintaining optimal performance across diverse network scenarios.}

	\subsubsection{Encoder}
	{In this section, we describe how the encoder maps user information $\bf{X}$ through the neural network into users' embedding $\bf{E}$ suitable for subsequent processing.}
	Users' information $\mathbf{X}=[\mathbf{x}_1,\dots,\mathbf{x}_N]$ is first be expanded into $d_e$ dimensions by a fully connected feed forward ($\text{FF}_1$) layer, and then it passes a multi-head attention (MHA) layer and a feed forward ($\text{FF}_2$) layer. Both MHA and $\text{FF}_2$ layer have a residual connection and are followed by batch normalization in \cite{kool2018attention}. Hence, we have:
	\begin{eqnarray}
		\begin{split}
			&\hat{\mathbf{E}} =\text{BN}\bigg(\text{FF}_1\big(\mathbf{X}\big)+\text{MHA}\Big(\text{FF}_1\big(\mathbf{X}\big)\Big)\bigg)  \\
			&\mathbf{E}=\text{BN}\bigg(\hat{\mathbf{E}}+\text{FF}_2\big(\hat{\mathbf{E}}\big)\bigg).
		\end{split}
	\end{eqnarray}
	The details of the multi-head attention mechanism are shown in Appendix~A.
	{ The length of users' embedding is adaptive to the variable number of users. Each term in the obtained users' embedding $\mathbf{E}=[\mathbf{e}_1,\dots,\mathbf{e}_N]$ takes into account all the users' information.}

	\subsubsection{Decoder}
	{The decoder iteratively generates the users' SIC ordering by utilizing the user embeddings obtained from the encoder.}
	In the decoder, the user embedding \( \mathbf{E} \) is initially passed through a linear layer to derive \( \mathbf{K} \), \( \mathbf{V} \), and \( \mathbf{K}' \) as follows:
	\begin{equation}
		\label{equation_fixedinput}
		\mathbf{K} = \mathbf{W}^{K}\mathbf{E}, \quad
		\mathbf{V} = \mathbf{W}^{V}\mathbf{E}, \quad
		\mathbf{K}' = \mathbf{W}^{{K}'}\mathbf{E},
	\end{equation}
	where \( \mathbf{W}^{K} \), \( \mathbf{W}^{V} \), and \( \mathbf{W}^{{K}'} \) are matrices of learnable parameters. {The global embedding \( \overline{\mathbf{e}} = \frac{1}{N}\sum_{n=1}^{N}\mathbf{e}_n \) is computed to effectively capture the overall network state. The derived values of \( \mathbf{K} \), \( \mathbf{V} \), and \( \mathbf{K}' \), along with \( \overline{\mathbf{e}} \), are then utilized in subsequent iterations of the decoder.}
	
	As shown in Fig.~\ref{fig:algorithm_overview}, the decoder performs in an iteration mode and generates an $N$ users' SIC ordering $\bm{\pi}=[\pi_1,\pi_2,...,\pi_N]$.
	For each iteration $t \in \left[1,2,...,N\right]$, the decoder takes the last decoded user's index $\pi_{t-1}$ and decides the current decoded user's index $\pi_{t}$.
	
	Firstly, the concatenation module concatenates an input vector $\left[\overline{\mathbf{e}}, \mathbf{e}_{\pi_{t-1}}\right]$  that contains the global information and the information of the previously decoded users.  $\mathbf{e}_{\pi_{t-1}}$ denotes the previous decoded user's embedding. {The embedding $\mathbf{e}_{\pi_{t-1}}$ of the input vector changes with iteration, which captures the user preferences on the SIC ordering.} When decoding the first user with no previous decoded user,  $\mathbf{e}_0$ is set as the learnable parameter vector. 
	
	Secondly, the masked multi-head attention layer generates a feature vector $\mathbf{u}_t$ that involves the query, keys, and values initially introduced in \cite{vaswani2017attention}.  The keys $\mathbf{K}$ and values $\mathbf{V}$ are obtained in (\ref{equation_fixedinput}), while the single query is variable input computed as follows: 
	\begin{equation}
		\label{equation_linear}
		\mathbf{q}_t=\mathbf{W}^{Q} \left[\overline{\mathbf{e}}, \mathbf{e}_{\pi_{t-1}}\right],
	\end{equation}
	where $\mathbf{W}^{Q}$ are learnable parameters matrices. The multi-head attention mechanism is described in Appendix A and omitted for brevity.
	Then the output of the masked multi-head attention layer can be calculated as
	\begin{equation}
		\mathbf{u}_t= \text{softmax}\left(\text{mask}\left(\tfrac{{\mathbf{q}_t}^{T}\mathbf{K}}{\sqrt{d_k}}\right)\right)\mathbf{V},
	\end{equation}
	where $d_k$ is the dimension of $\mathbf{q}$, the softmax function can refer to equation~(4) of Appendix~A, and the mask function can be expressed as:
	\begin{equation}
		\text{mask}\left(\tfrac{{\mathbf{q}_t}^{T}\mathbf{k}_i}{\sqrt{d_k}}\right) = \begin{cases}
			-\infty & \text{ if } i\in\{\pi_1,\pi_2,...,\pi_{t-1}\}, \\
			\frac{{\mathbf{q}_t}^{T}\mathbf{k}_i}{\sqrt{d_k}} & \text{ otherwise },
		\end{cases}
	\end{equation}
	where $\mathbf{k}_i$ denotes the $i$-th element of the keys $\mathbf{K}$.
	At each iteration $t$, the users selected in previous iterations are masked to guard that each user's index appears precisely once in $\bm{\pi}$.
	
	
	Thirdly, the single-head attention layer is used to generate the probability of users being selected at time $t$. $\bm{\ell}_{t}$ can also be considered the similarity between the single query and the keys of the single-head attention layer. The keys $\mathbf{K}'$ is obtained in (\ref{equation_fixedinput}), and the single query $\mathbf{q}'_t$ is computed as follows:
	\begin{equation}
		\mathbf{q}'_t=\mathbf{W}^{Q'} \mathbf{u}_t,
	\end{equation}
	where $\mathbf{W}^{Q'}$ are learnable parameters matrices. Then, $\bm{\ell}_{t}$ can be derived as
	\begin{equation}
		\label{Eqn:probability}
		\bm{\ell}_{t}= \text{softmax}\left(\text{mask}\left(\text{clip}\left(\tfrac{{\mathbf{q}'_t}^{T}\mathbf{K}'}{\sqrt{d_k}}\right)\right)\right),
	\end{equation}
	where the clip function can clip the result within $\left[-10,10\right]$ by the tanh function to avoid the probability of each selected user being too large or too small \cite{kool2018attention}. 
	
	Finally, the selection module selects the user to be decoded based on $\bm{\ell}_{t}$. In the inference phase, the selection module works in the greedy mode. Specifically, the selection module greedily selects the one with the greatest probability to be the $t$-th decoded user, as
	\begin{equation}\label{equation_find_max_pro_user}
		\pi_t=\mathop{\arg\max}_{n}\ell_{t,n}.
	\end{equation}
	Up to now, the decoder operation at iteration $t$ is completed. 
	
	Upon repeating the aforementioned steps \( N \) times in the decoder, we can compile all \( \pi_t \) to form the SIC decoding order \( \bm{\pi} \):
	\begin{equation}
		\bm{\pi}=[\pi_1,\dots,\pi_N].
	\end{equation}

	\subsubsection{Step-by-Step SIC Ordering Example}
	\begin{figure}
		
		\centering
		\includegraphics[width=1 \linewidth]{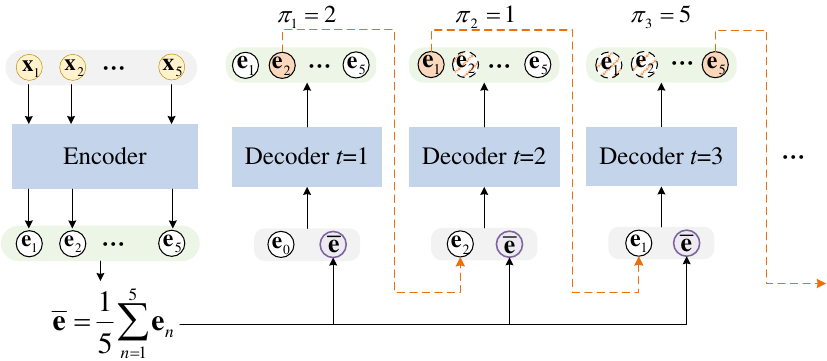}
		\caption{The process of iteratively determining the SIC ordering in ASOPA.}
		\label{fig:Aglorithm_Process}
	\end{figure}
	{
	To illustrate the iterative generation of the SIC ordering, we provide an example involving 5 users as shown in Fig.~\ref{fig:Aglorithm_Process}. 
	
	The encoder takes the five users' information $\mathbf{X}=[\mathbf{x}_1,\mathbf{x}_2,\dots,\mathbf{x}_5]$ as input, and correspondingly calculate the users' embedding $\mathbf{E}=[\mathbf{e}_1,\mathbf{e}_2,\dots,\mathbf{e}_5]$ and the global embedding ${\bf{\bar e}} = \frac{1}{5}\sum_{n = 1}^5 {{{\bf{e}}_n}} $.
	
	The decoder then iteratively generates the first three SIC orderings of five users as follows:
	\begin{enumerate}
	\item 
	In the first iteration (\(t=1\)), the decoder inputs \([\overline{\mathbf{e}}, \mathbf{e}_{0}]\), computes the probabilities \(\bm{\ell}_{t}\) using Equation~\eqref{Eqn:probability}, and selects the user with the highest probability for the first decoding. In this example, user~2 has the highest probability (\(\ell_{1,2}=0.4187\)) and is selected as \(\pi_1 = 2\). 
	\item
	In the second iteration (\(t=2\)), the decoder takes the global embedding \(\bar{\mathbf{e}}\) and the embedding of the first decoded user (\(\mathbf{e}_2\)) as inputs \([\overline{\mathbf{e}}, \mathbf{e}_{2}]\), and recalculates the probabilities \(\bm{\ell}_{t}\). As user~2 has already been selected, its probability is masked and set to zero (\(\ell_{2,2}=0\)). The highest probability in this iteration is \(\ell_{2,1}=0.4728\), leading to the selection of user~1 as \(\pi_2 =1\). 
	\item
	In the third iteration (\(t=3\)), the decoder inputs \([\overline{\mathbf{e}}, \mathbf{e}_{1}]\), and recalculates the probabilities \(\bm{\ell}_{t}\).
	With the probabilities of the previous selected users masked ($\ell_{3,2}, \ell_{3,1} =0 $), user~5 has the highest probability (\(\ell_{3,5}=0.7953\)) in this iteration and thus user~5 is selected as $\pi_3=5$.
	\end{enumerate}
	This procedure is repeated for two more iterations until the SIC decoding order for all five users is determined.
	The specific probabilities \({\ell}_{t,n}\) of all users over iterations are shown in Table~\ref{tab:case_decoder}.

	 }
	{
	\begin{table*}
		\centering
		
		\caption{Case study - Inference of the decoder}
		\label{tab:case_decoder}
		\begin{tabular}{c|c|lllll|c}
			\hline
			&{\multirow{3}{*}{Input}} & \multicolumn{6}{c}{Output}\\
			\cline{3-8}
			&&\multicolumn{5}{c|}{$\bm{\ell}_{t}$}& \multirow{2}{*}{Decoded user $\pi_t$}\\
			&&user 1& user 2& user 3&user 4&user 5&\\
			\hline
			$t=1$& $\left[\overline{\mathbf{e}}, \mathbf{e}_{0}\right]$&0.2727& \bfseries{0.4187}& 0.0223& 0.0295& 0.2568& $\pi_1 = 2$ \\
			
			$t=2$&  $\left[\overline{\mathbf{e}}, \mathbf{e}_{2}\right]$&  \bfseries{0.4728}& 0  (masked)& 0.0399& 0.0528& 0.4345& $\pi_2 = 1$ \\ 
			$t=3$& $\left[\overline{\mathbf{e}}, \mathbf{e}_{1}\right]$ & 0  (masked)&0  (masked)& 0.0893&0.1154&\bfseries{0.7953}& $\pi_3 = 5$\\
			
			$t=4$& $\left[\overline{\mathbf{e}}, \mathbf{e}_{5}\right]$ & 0  (masked)& 0  (masked)&0.4520& \bfseries{0.5480}&0  (masked)& $\pi_4 = 4$\\
			
			$t=5$& $\left[\overline{\mathbf{e}}, \mathbf{e}_{4}\right]$& 0  (masked) & 0  (masked) & \bfseries{1} & 0  (masked)& 0  (masked)& $\pi_5 = 3$\\
			\hline
		\end{tabular}
	\end{table*}
	}

	\subsection{Power Allocation Under Given SIC Ordering}\label{sec_convex_opt}
	
	{
	In this subsection, we design a convex transformation algorithm for the power allocation problem under given SIC ordering. Referring to Fig.~\ref{fig:algorithm_overview}, the actor network of ASOPA only generates SIC ordering without considering transmit power allocation and the corresponding constraints. To obtain the corresponding power allocation and evaluate the SIC ordering, we solve the power allocation subproblem and obtain the achieved network utility as follows.
	}

	{
		Upon determining the SIC ordering \( \bm{\pi} \), the sub-problem \textbf{P1} is  addressed to identify the optimal power allocation satisfying the constraints and subsequently calculate the network utility \( R(\bm{\pi}|\mathbf{X}) \) for the specified SIC ordering. 
		Given that the data rate \( R_n \) for each user is non-convex, \textbf{P1} is inherently a non-convex problem. To tackle this, we employ variable substitution to transform \textbf{P1} into a convex problem. The details of this transformation and the proof of its convexity are provided in Appendix B. We solve the transformed version of \textbf{P1} using the interior-point method \cite{boyd2004convex} of the CVX solver. The solution of \textbf{P1} yields the optimal power allocation $\mathbf{p}$ under the given SIC ordering \( \bm{\pi} \) satisfying the constraints. Following this, ASOPA outputs the network decision \( (\bm{\pi}, \mathbf{p}) \) based on the instantaneous user information \( \mathbf{X} \). 	}
	
		Consequently, the network utility \( R(\bm{\pi}|\mathbf{X}) \) can be calculated. 
		{	This calculated utility provides essential feedback on the effectiveness of the SIC ordering under the current policy. As such, the resource allocation module acts as a critic in training the actor network, playing a crucial role in evaluating these generated SIC orderings.
	
		Notice that ASOPA can be migrated to any other NOMA optimization problem with the required resource allocation problem. The extensibility of ASOPA is discussed in following Sec.~\ref{sec:Extension}.
	}
	
	{
		\subsection{Policy Update}	
		The red dotted lines in Fig.~\ref{fig:algorithm_overview} represent the policy update process.}	For an input instance $\mathbf{X}$, our goal is to maximize the expected sum of weighted users' logarithmic throughput, as
	\begin{equation}\label{equation_expecatation_R}
		\mathbb{E}_{\bm{\pi}\sim z_{{\theta}}\left(\cdot|\mathbf{X}\right)}R\left(\bm{\pi}|\mathbf{X}\right).
	\end{equation}
{	
	where $z_\theta\left(\bm{\pi}|\mathbf{X}\right)=\prod_{t=1}^{N}\ell_{t,\pi_t}$ is a measure of the likelihood that the generated SIC decoding order \( \bm{\pi} \) is the optimal sequence given the current network state \( \mathbf{X} \).}
	
	To explore more SIC orderings in the training phase, at the selection module of the decoder network, the user based on the probability $\bm{\ell}_t$ at the iteration $t$ is sampled \cite{kool2018attention} as
	\begin{equation}\label{equation_sample_user}
		\pi_t \sim \bm{\ell}_t.
	\end{equation}
	After iterating $t \in \left[1,2,...,N\right]$, we obtain the SIC ordering $\bm{\pi}$. {Then, the gradient of the actor network $\theta$ can be formulated by the REINFORCE algorithm \cite{williams1992simple} as
	\begin{equation} \label{Eqn:Rgradient}
		\begin{split}
			\mathbb{E}\Big(\!R\left(\bm{\pi}|\mathbf{X}\right)\!\nabla_{{\theta}}\log z_{{\theta}}\left(\bm{\pi}|\mathbf{X}\right) \Big).
		\end{split}
	\end{equation}
	However, the REINFORCE algorithm may be of high variance and thus produce slow learning \cite{SuttonRS2018}. To improve the performance of DRL, the REINFORCE algorithm with baseline \cite{SuttonRS2018} is adopted in this paper.
	As illustrated in Fig.~\ref{fig:algorithm_overview}, the baseline network $\theta'$, which uses the actor network parameters from the previous epoch, serves as a baseline to generate SIC orderings $\bm{\pi}^{BL}$ and subsequently calculates the network utility $R(\bm{\pi}^{BL}|\mathbf{X})$. 
	The difference between $R(\bm{\pi}^{BL}|\mathbf{X})$ and  $R(\bm{\pi}^{BL}|\mathbf{X})$ is utilized to train the current actor network $\theta$.
	
	Consequently, the gradient of the REINFORCE algorithm with baseline can be expressed and then approximated by Monte Carlo sampling as
		\begin{equation} \label{Eqn:gradient}
			\begin{split}
			&\mathbb{E}\left(\left(\!R\left(\bm{\pi}|\mathbf{X}\right)\!\!-\!\! R(\bm{\pi}^{BL}|\mathbf{X})\!\right)\!\nabla_{{\theta}}\log z_{{\theta}}\left(\bm{\pi}|\mathbf{X}\right)\right)\\
			&\approx \frac{1}{|\bm{\tau}|}\sum_{i=1}^{|\bm{\tau}|}\bigg(\!\!\!\left(R\left(\bm{\pi}_i|\mathbf{X}_i\right)- R\left(\bm{\pi}_i^{BL}|\mathbf{X}_i\right)\right)\nabla_{{\theta}}\log z_{{\theta}}\left(\bm{\pi}_i|\mathbf{X}_i\right)\!\!\bigg).
			\end{split}
		\end{equation}
	where $|\bm{\tau}|$ is the batch size, $\mathbf{X}_i$ is the $i$-th input, $\bm{\pi}_i$ is the SIC ordering  produced by the actor network based on \eqref{equation_sample_user}, and $\bm{\pi}_i^{BL}$ is the $i$-th SIC ordering produced by the baseline network.}
	In practice, the expectation in Equation~\eqref{Eqn:gradient} is approximated by averaging over a batch of uniformly sampled input instances \( {\left\{ {{{\bf{X}}_\tau }} \right\}_{\tau  \in \bm{\tau} }} \), where \( \bm{\tau} \) represents the index set of these sampled instances.
	After obtaining the gradients (\ref{Eqn:gradient}), Adam \cite{kingma2014adam} is applied as the optimizer to update the actor network's parameters ${\theta}$.

	{Our reinforcement learning algorithm for training the actor network is outlined in Algorithm~\ref{algorithm_train}. To facilitate this process, an empty memory with limited capacity is established to store past samples. As new samples are received in each time slot, policy updates are executed infrequently. For every policy update, a random batch of samples is selected from this memory to train the actor network. The baseline network, on the other hand, undergoes updates every epoch which consists of $M$ times policy update. This update process involves copying the parameters from the actor network to the baseline network, as denoted by \( \theta'=\theta \). This systematic approach ensures continuous adaptation and optimization of the actor network's performance based on the latest data.}

				%

	\begin{algorithm}[h]\label{algorithm_train}
		\SetAlgoLined 
		
		\SetKwData{Left}{left}\SetKwData{This}{this}\SetKwData{Up}{up}
		\SetKwRepeat{doWhile}{do}{while}
		\SetKwFunction{Union}{Union}\SetKwFunction{FindCompress}{FindCompress}
		\SetKwInOut{Input}{input}\SetKwInOut{Output}{output}
		\caption{Training ASOPA}
		\Input{Users' weights, maximum transmit power, and channel gains at each time slot $s$ $\mathbf{X}_{s}=\{ (w_i,P_i^{max}, g_i)\}_{i\in \mathcal{N}}$ , the training interval $\delta_T$ of the actor network, the update epoch $\delta_E$ of the baseline network;}
		\Output{SIC order $\bm{\pi}$ and power allocation $\mathbf{p}$;}
		Initialize the actor network's parameters ${\theta}$; \\
		Initialize the baseline network's parameters ${\theta}'\leftarrow {\theta}$; \\
		\For{$epoch=1,\cdots,E$}{
			Generate the SIC ordering $\bm{\pi}$ for $\mathbf{X}$ of the epoch from the actor network based on \eqref{equation_sample_user}; {//Inference of ASOPA for each sample}   \\
			Obtain $\mathbf{p}$ and $R\left(\bm{\pi}|\mathbf{X}\right)$ for $\mathbf{X}$ of the epoch by solving problem  \textbf{P1}; \\
			\For{$batch=1,\cdots,M$}{
				Uniformly sample a batch of samples $\{\bm{X}_{\tau}\}_{\tau\in\bm{\tau}}$ from prvious samples of the epoch; {//Infrequently policy update of ASOPA can be executed in parallel on different servers in practical applications}\\
				Generate the SIC ordering $\{\bm{\pi}_{\tau}\}_{\tau\in\bm{\tau}}$ from the actor network based on \eqref{equation_sample_user};\\
				Generate the baseline SIC ordering $\{\bm{\pi}^{BL}_{\tau}\}_{\tau\in\bm{\tau}}$ from the baseline network based on \eqref{equation_find_max_pro_user};\\
				Obtain $\{\left(R\left(\bm{\pi}_{\tau}|\mathbf{X}_{\tau}\right), R\left(\bm{\pi}_{\tau}^{BL}|\mathbf{X}_{\tau}\right)\right)\}_{\tau\in\bm{\tau}}$ by solving problem \textbf{P1};\\
				Calculate the gradient $\nabla_{{\theta}} S(\theta)$ from (\ref{Eqn:gradient}) based on $\{\left(R\left(\bm{\pi}_{\tau}|\mathbf{X}_{\tau}\right), R\left(\bm{\pi}_{\tau}^{BL}|\mathbf{X}_{\tau}\right)\right)\}_{\tau\in\bm{\tau}}$;\\
				Update the actor network's parameters $\theta$ using the Adam optimization algorithm based on the calculated gradients $\nabla_{{\theta}}$;
			}
			Update the baseline network's parameters ${\theta}'\leftarrow{\theta}$;
		}
	\end{algorithm}
	
	{	
		\subsection{Computation Complexity}
		ASOPA operates through two distinct processes: the inference process and the policy update process. During each time slot, ASOPA's inference process is activated to generate Successive Interference Cancellation (SIC) orderings and power allocations. Contrarily, the policy update process can be carried out less frequently, and can also be executed in parallel on different servers in practical applications. Given the crucial role of inference delay in determining the feasibility of field deployment, the inference complexity of ASOPA is a key area of interest.
		
		The inference process is detailed in lines 4-5 of Algorithm~\ref{algorithm_train}. Line 4 involves generating the SIC ordering from the actor network, where the computational complexity is primarily driven by matrix multiplications in the attention mechanism. The complexity of interactions between the query, key, and value is \( O(HDN^2) \), with \( H \) representing the number of heads in multiple attention mechanisms, \( D \) the dimension of these components, and \( N \) the number of users. Line 5 addresses the generation of power allocation by solving subproblem \textbf{P1}, which is reformulated into a convex problem \textbf{P2} and solved using the cvxopt solver with the Interior Point Method. The computational complexity of this method is \( O(N^{3.5}) \) \cite{MeS1992convex}. Therefore, the overall inference complexity of ASOPA is \( O(N^{3.5}) \). In Section~\ref{sec:latency}, we will numerically demonstrate that ASOPA meets the real-time requirements of NOMA networks.
	}

	{
	\section{Extension Scenarios} \label{sec:Extension}
	{
	ASOPA can be easily extended to various NOMA scenarios. To migrate to a new scenario, only the inputs of the actor network and baseline network need to be modified, while the structure of the actor network used to determine the SIC ordering remains unchanged. Correspondingly, the resource allocation subproblems are adjusted and re-solved for specific NOMA scenarios, ensuring optimal performance and efficiency under different network conditions.
	}

	In this section, we evaluate the extension scenarios of ASOPA, including NOMA networks with imperfect CSI, NOMA networks with multiple-antenna setups, and NOMA networks with multiple-BS setups.
	
	\subsection{Multiple Antenna}
	\begin{figure}
		\centering
		\includegraphics[width=0.7 \linewidth]{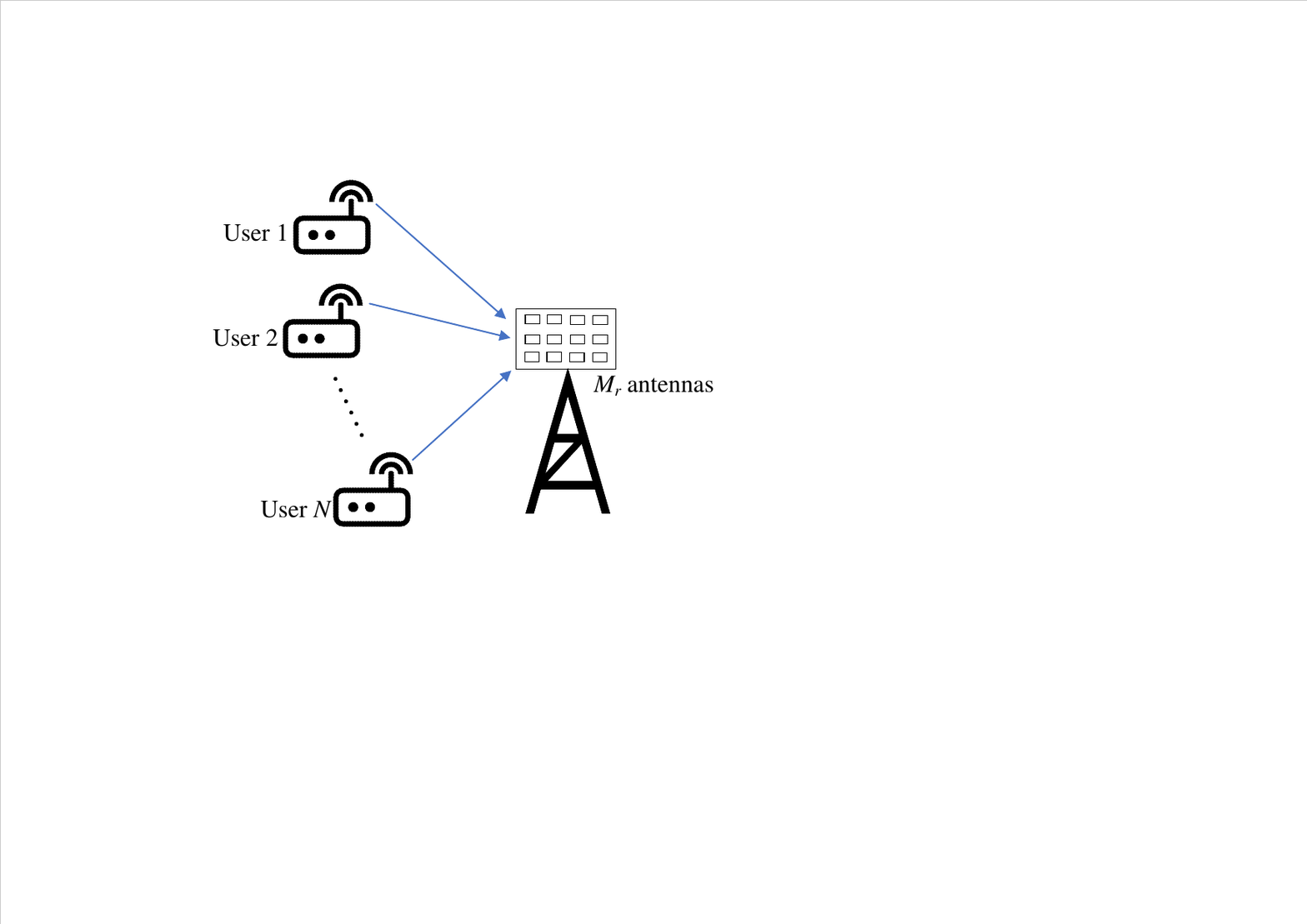}
		\caption{A NOMA network with multiple-antenna setup at the BS.}
		\label{fig:system_model_multipleantennas}
	\end{figure}
	In this subsection, we evaluate ASOPA in NOMA networks with multiple-antenna setups, since multiple-antenna technology has advantages in improving spectrum and energy efficiency \cite{QiX2023,ZhangJ2022}. 
	The uplink multiple-antenna system consists of a ${M_r}$ antennas BS and $N$ single-antenna users as shown in Fig.~\ref{fig:system_model_multipleantennas}. The received signal  ${\bf{Y}} \in {\mathbb{C}^{{M_r} \times 1}}$ of BS is
	\begin{equation}
		{\bf{Y}} = {\bf{HS}} + {\bf{n}},
	\end{equation} 
	where ${\bf{H}} \in {\mathbb{C}^{{M_r} \times N}}$  denotes the channel state from users to the receive antennas of BS, ${\bf{S}} = {\left[ {{s_1},{s_2},...,{s_N}} \right]^T}$  denotes the transmit signal matrix, and ${\bf{n}}\sim{\cal{CN}}\left( {{\bf{0}},{\sigma ^2}{\bf{I}}} \right)$ represents the additive white Gaussian noise at the BS side. ${\cal{CN}}(0,\sigma^2)$ denotes the complex Gaussian distribution with mean zero and the variance $\sigma^2$. The linear equalization, such as zero-forcing (ZF) or minimum-mean-square-error (MMSE), is used in multiple-antenna scenarios to symbol-by-symbol detection. According to the states of multiple-antenna scenario, ASOPA correspondingly generates the SIC ordering and power allocation. The specific steps are as follows.
	
	The equalization matrix ${\bf{V}}\in {\mathbb{C}}^{{N} \times {M_r}}$ for ZF \cite{QiX2023} or MMSE \cite{ZhangJ2022} can be expressed as \cite{ChrW2004}
	\begin{equation}
		{\bf{V}} = \left\{ {\begin{array}{*{20}{l}}
				{\bf{P}}{{\bf{H}}^H}{{({{\bf{H}}}{\bf{P}}{\bf{H}}^H + \sigma {\bf{I}})}^{ - 1}} &{\text{when MMSE}},\\
				{{\bf{H}}^H}{{({{\bf{H}}}{\bf{H}}^H)}^{ - 1}} &{\text{when ZF}},
		\end{array}} \right.
	\end{equation}
	where ${\bf{P}}$ denotes the diagonal matrix ${diag}(p_1, p_2, …, p_N)$.
	
	Then the received estimated signal is
	\begin{equation}
		{\bf{\hat S}} = {\bf{VY}} = {\bf{VHS}} + {\bf{Vn}},
	\end{equation}
	and the estimated value of user $n$ is
	\begin{equation}
		{\hat s_n} = \sqrt {{p_n}} {{\bf{v}}_n}{{\mathbf{h}}_n}{s_n} + \sum\limits_{n' \ne n} {\sqrt {{p_{n'}}} } {{\bf{v}}_n}{{\bf{h}} _{n'}}{s_{n'}} + {{\bf{v}}_n}{\bf{n}},
	\end{equation}
	where ${{\bf{h}}_n} \in {\mathbb{C}^{{M_r} \times 1}}$ denotes the channel states from user $n$ to the receive antenna of the BS, and ${\bf{v}}_n\in {\mathbb{C}^{1 \times {{M_r}}}}$ denotes the $n$-th row of ${\bf{V}}$, which can be given by
	\begin{equation}
		\label{eqn:MMSEVn}
		{\bf{v}}_n\!\!= \!\!\left\{ \!{\begin{array}{*{20}{c}}\!\!\!
				{{p}}_n{{\bf{h}}_n^H}{{\left(\!\sum\limits_{\xi\left(n'\right)\ge\xi\left(n\right),\forall n' \in\mathcal{N}}\!\!\!\!\!\!\!\!\!\!\!\!{p_n}{{\bf{h}}_{n'}}{{\bf{h}}_{n'}^H} \!+ \!\sigma {\bf{I}}\!\right)}^{\!- 1}} \!\!\!\!\!\!\!\!&\!\!\!\!{\text{when MMSE}},\!\!\!\\
				{{\bf{h}}_n^H}{{({{\bf{h}}_n}{\bf{h}}_n^H)}^{ - 1}} \!\!\!&\!\!\!\!\!\!\!\!\!\!\!{\text{when ZF}},\!\!
		\end{array}} \right.\!\!\!
	\end{equation}
	
	According to the estimated value of user $n$, its achieved transmit rate can be calculated by
	\begin{equation}
		{R_n} = {\log _2}\left( {1 \!+ \!\frac{{{{\left| {{{\bf{v}}_n}{{\bf{h}}_n}} \right|}^2}{p_n}}}{{\sum\limits_{\xi\left(n'\right)>\xi\left(n\right),\forall n' \in\mathcal{N}}\!\!\!\!\!\!\!\!\!\! {{{\left| {{{\bf{v}}_n}{{\bf{h}}_{n'}}} \right|}^2}{p_{n'}} + } {{\left| {{{\bf{v}}_{n}}} \right|}^2}{\sigma ^2}}}} \right).
	\end{equation}

	When using ZF method, the equalization matrix $\bf{V}$ in (\ref{eqn:MMSEVn}) is independent of $\bf{p}$, so that the power allocation problem can be transformed into a convex as appendix~E and solved by the CVX solver.
	
	
	However, when using the MMSE method, the equalization matrix $\bf{V}$ depends on the variable $\bf{p}$, making the power allocation problem of MMSE intractable and non-convex. 
	To address this difficulty, alternative optimization is employed to decompose the non-convex problem into two subproblems: one for $\bf{V}$ and one for $\bf{p}$. When $\bf{V}$ is fixed, the power allocation problem can be transformed into a convex problem, as shown in Appendix E. Therefore, the alternative optimization starts with an initial power allocation to calculate the corresponding equalization matrix. Using this equalization matrix, it then applies the convex method to determine a new power allocation. With this updated power allocation, the equalization matrix is recalculated. This process repeats until the difference between successive power allocations is smaller than the set threshold.

	Different from the single antenna scenario, the channel gain between the BS and a user $n$ is a complex value, whose real and imaginal parts are denoted as $\bm{h}^r_{n}=\{h^r_{1,n},...,h^r_{M_r,n}\}$ and $\bm{h}^c_{n}=\{h^c_{1,n},...,h^c_{M_r,n}\}$, respectively. To tackle multiple-antenna scenarios, ASOPA modifies its input from ${\bf{X}} = {\{ ({w_n},P_i^{max},{g_n})\} _{n \in {\cal N}}}$ to ${{\bf{X}}} = {\{ ({w_n},P_n^{max},{\bm{h}^r_{n}},\bm{h}^c_{n})\} _{n \in {\mathcal N}}}$. The rest of the ASOPA structure remains the same as one illustrated in Fig.~\ref{fig:algorithm_overview}.
}
	{
	\subsection{Imperfect Channel}
	In this subsection, we assess the impact of estimation errors on the performance of ASOPA. 
	For the perfect CSI scenario, the channel gain is expressed as \(g_n=\overline{g}_n|\alpha_n|^2\), where \(\overline{g}_n=A_d\left(\frac{3\cdot10^8}{4\pi f_c b_n}\right)^{b_e}\) and $\alpha_n~\sim~{\cal{CN}}\left( {0,1} \right)$ account for path loss power gain and the Rayleigh fading channel coefficient between BS and $n$-th user, respectively.
	Since the path loss coefficient are large-scale fading factors and are slowly varying, we assume that the path loss coefficient $\overline{g}_n$ between BS and each user can be estimated perfectly.
	However, in dynamic and complex wireless environments, accurately acquiring time-varying Rayleigh fading channel gains is challenging. Following the approaches in \cite{YooT2006,HanS2009}, the Rayleigh fading channel gain is modeled as
	\begin{equation}
		{\alpha_n} = {\hat \alpha_n} + {\epsilon_n}
	\end{equation}
	where ${\alpha_n}$ is the realistic Rayleigh fading channel coefficient between BS and $n$-th user, ${\hat h_n} \sim {\cal{CN}}\left( {0, 1-\sigma_{\epsilon}^2} \right)$ denotes the estimated channel coefficient, and ${\mathcal{\epsilon}_n}~\sim~{\cal{CN}}\left( {0,\sigma _{{{\mathcal{\epsilon}}}}^2} \right)$ is the estimated error. Note that the parameter  $\sigma _{{{\mathcal{\epsilon}}}}^2$ indicates the quality of channel estimation, and keeps constant as \cite{YangZ2016,GaoY2018}. We assume that ${\hat \alpha_n}$ and ${\epsilon_n}$ are uncorrelated.
	
	If the perfect CSI is known, the maximum achievable data rate between BS and $n$-th can be written as
	\begin{equation}
		c_n = W \log_2(1+\phi_n)
	\end{equation}
	where
	\begin{equation}
		{\phi _n} = \frac{{{p_n}|\alpha_n|^2{{\overline g}_n}}}{{\sum\limits_{\xi \left( {n'} \right) > \xi \left( n \right),\forall n' \in {\cal N}} {{p_{n'}}|\alpha_n|^2{{\overline g}_{n'}}}  + {N_0}}}. \label{eqn:PerfectSINR}
	\end{equation}
	In (\ref{eqn:PerfectSINR}), $\phi _n$ denotes the signal-to-interference-plus-noise ratio (SINR) of the user $n$. In practice, the BS can only obtain the estimated fading channel coefficient ${\hat \alpha}_n$. The scheduled data rate with imperfect CSI can be expressed as
	\begin{equation}
		r_n = W\log_2(1+{\hat \phi}_n)
	\end{equation}
	where
	\begin{equation}
		{{\hat \phi}_n} = \frac{{{p_n}|\hat \alpha_n|^2{{\overline g}_n}}}{{\sum\limits_{\xi \left( {n'} \right) > \xi \left( n \right),\forall n' \in {\cal N}} {{p_{n'}}|\hat \alpha_n|^2{{\overline g}_{n'}}}  + {N_0}}}. \label{eqn:ImperfectSINR}
	\end{equation}
	However, the scheduled data rate with imperfect CSI may easily exceed the maximum achievable data rate, i.e., $r_n>c_n$. To measure the performance of this case, we introduce outage probability as a metric \cite{FangF,ZhangH2020}. Therefore, the weighted proportional fairness function with outage probability can be expressed as $\sum\limits_{n=1}^{N}{w_n\ln r_n \Pr[r_n\le c_n|{\hat \alpha}_n]}$.
	$\Pr[r_n\le c_n|{\hat \alpha}_n]$ denotes the probability of a case when the scheduled data rate $r_n$ is less than or equal to the maximum data rate $c_n$ under the estimated channel coefficient ${\hat \alpha}_n$. The optimization problem can be reformulated as
	\begin{subequations}
		\label{Eqn:ImperfectProb}
		\begin{align}
			\max\limits_{\mathbf{\bm{\pi}},\textbf{p}}\ &\sum\limits_{n=1}^{N}{w_n\ln r_n \Pr[r_n\le c_n|{\hat \alpha}_n]} \label{ImperfectCSIobjective}\\
			s.t.\ & \Pr[c_n<r_n|{\hat \alpha}_n]\le \epsilon_{out}, \forall n \in\mathcal{N},\label{Outage_constraint} \\
			&0<p_n\le P_n^{max},\forall n \in\mathcal{N},\\
			&\bm{\pi}\in\Pi,
		\end{align}
	\end{subequations}
	where (\ref{Outage_constraint}) is introduced to satisfy the channel outage probability requirement $\epsilon_{out}$ for all users in the imperfect CSI scenario.
	Due to the probability constraints (\ref{Outage_constraint}), this problem (\ref{Eqn:ImperfectProb}) turns into a non-convex problem and cannot easily be optimally solved in polynomial time \cite{FangF}. To tackle this problem efficiently, we transform the probabilistic mixed problem into a non-probability problem as
	\begin{subequations}
		\label{TransformedProblem}
		\begin{align}
			\max\limits_{\mathbf{\bm{\pi}},\textbf{p}}\ &\sum\limits_{n=1}^{N}{w_n\ln(1-\epsilon_{out}) {\tilde r}_n}\\
			s.t.\ &0<p_n\le P_n^{max},\forall n \in\mathcal{N},\\
			&\bm{\pi}\in\Pi,
		\end{align}
	\end{subequations}
	where ${\tilde r}_n = W\log_2(1+{\tilde \phi}_n)$, and the transformed SINR ${\tilde \phi}_n$ can be expressed as
	\begin{equation}
			{\tilde \phi}_n = \frac{\epsilon_{out}{F_{|g_n|^2}^{-1}(\epsilon_{out}/2)p_n }}{\epsilon_{out}\sigma_{\epsilon}^2+\!\!\!\!\!\!\!\!\!\!\sum\limits_{\xi \left( {n'} \right) > \xi \left( n \right),\forall n' \in {\cal N}}\!\!\!\!\!\!\!\!\!\!2(|{\hat g}_{n'}|^2+\sigma^2_{\epsilon})p_{n'}},
	\end{equation}
	where $F_{|g_n|^2}^{-1}(\epsilon_{out}/2)$ denotes the inverse cumulative distribution function of a noncentral chi-square random variable with 2 degrees of freedom and non-centrality parameter $2|{\hat g}_{n}|^2/\sigma^2_{\epsilon}$. The details of the probabilistic mixed problem transformation are shown in Appendix~D.
	
	
	Notice that once the SIC ordering is determined, the power allocation problem of (\ref{TransformedProblem}) can be transformed into a convex problem as well as Appendix~B and solved by the CVX solver. Thus, ASOPA can be applied to solve it by simply modifying its input from ${\bf{X}} = {\{ ({w_n},P_i^{max},{g_n})\} _{n \in {\cal N}}}$ to ${{\bf{X}}} = {\{ ({w_n},P_n^{max},{|{{\hat \alpha}_n|}^2{\overline g}_n})\} _{n \in {\mathcal N}}}$ }
	
	{
	\subsection{Multiple BS}
	In this subsection, we assess the performance of ASOPA in NOMA networks with multiple BSs, represented by the set \( \mathbf{B} \), each containing \( N_b \) users where \( b \in \mathbf{B} \). In scenarios involving multiple BSs, each user experiences inter-cell interference from users linked to other BSs. To quantify this interference, the channel gain from a user \( n \) in BS \( b\) to another BS \( b' \) is defined as \( h^{({b})}_{n,b'} \). Specifically, the superscript \( (b) \) indicates that user \( n \) is associated with BS \( b \).
	
	To accommodate the multiple-BS scenario, it is straightforward to adjust the input of ASOPA to ${\bf{X}} = \{(w_n^{(b)},P_{n,max}^{(b)},g_n^{(b)},\bm{h}^{(b)}_n)\}_{n\in\mathcal{N}_b,b\in\mathbf{B}}$, where \( \bm{h}^{(b)}_n = \{h^{(b)}_{n,b'}\}_{b'\in\mathbf{B}\setminus \{b\}} \). The resource allocation problem is still convex and can be solved using the CVX solver. The detailed setup of the multiple-BS scenario is in Appendix~E.
	
	The addition of inter-cell interference, however, adds complexity to the SIC ordering problem. In ASOPA, the next decoding user is iteratively chosen based on the highest probability from Equation~\eqref{equation_find_max_pro_user}, under the assumption that all users are within the same BS. 
	However, in scenarios involving multiple BSs, the SIC decoding order is specific to each BS. Comparing the probabilities of users from different BSs lacks physical insight, as such comparisons are not meaningful in this context.
	To overcome this issue, we introduce enhance the decoding process by introducing an additional masking mechanism in the decoder. This mechanism allows for the generation of appropriate SIC orderings for users across all BSs. Specifically, the decoder generates the SIC ordering for users associated with the current BS while effectively masking the users of other BSs. This approach ensures efficient decoding order determination in multi-BS NOMA networks without needing to modify Equation (\ref{equation_find_max_pro_user}).

	\begin{table*}
		\centering
		
		\renewcommand{\arraystretch}{1}
		\caption{Case study - The mask mechanism in ASOPA in a dual-BS NOMA network}
		\label{tab:case_decoder_mutipleBS}
		{
			\begin{tabular}{c|c|lll|lll|c}
				\hline
				&{\multirow{3}{*}{Input}} & \multicolumn{7}{c}{Output}\\
				\cline{3-9}
				&&\multicolumn{3}{c}{$\bm{\ell}^{(1)}_{t}$}&\multicolumn{3}{c|}{$\bm{\ell}^{(2)}_{t}$}&\multirow{2}{*}{Decoded user $\pi^{(b)}_t$}\\
				&&user 1& user 2& user 3&user 1&user 2& user 3 & \\
				\hline \\[-1em]
				$t=1$& $[\overline{\mathbf{e}}, \mathbf{e}_{0}]$& \bfseries{0.3997}& 0.2890&0.3114& 0(masked)& 0 (masked)& 0 (masked)&$\pi^{(1)}_{1} = 1$ \\[0.2em]
				
				$t=2$&  $[\overline{\mathbf{e}}, \mathbf{e}^{(1)}_{1}]$&   0  (masked)& 0.4839& \bfseries{0.5161}& 0  (masked)& 0  (masked)& 0  (masked)&$\pi^{(1)}_{2} = 3$ \\[0.2em]
				
				$t=3$& $[\overline{\mathbf{e}}, \mathbf{e}^{(1)}_{3}]$ & 0  (masked)&\bfseries{1}& 0  (masked)&0  (masked)&0  (masked)& 0  (masked)&$\pi^{(1)}_{3} = 2$\\[0.2em]
				
				$t=4$& $[\overline{\mathbf{e}}, \mathbf{e}^{(1)}_{2}]$ & 0  (masked)& 0  (masked)&0  (masked)& 0.0731& \bfseries{0.4698}& 0.4571&$\pi^{(2)}_{1} = 2$\\[0.2em]
				
				$t=5$& $[\overline{\mathbf{e}}, \mathbf{e}^{(2)}_{2}]$& 0  (masked) & 0  (masked) & 0  (masked) &0.0407 &0  (masked) &\bfseries{0.9593}&$\pi^{(2)}_{2} = 3$\\[0.2em]
				$t=6$& $[\overline{\mathbf{e}}, \mathbf{e}^{(2)}_{3}]$& 0  (masked) & 0  (masked) & 0  (masked) & \bfseries{1}& 0  (masked)& 0  (masked)&$\pi^{(2)}_{3} = 1$\\[0.2em]
				\hline
		\end{tabular}}
	\end{table*}
	\begin{figure}
		\centering
		\includegraphics[width=1\linewidth]{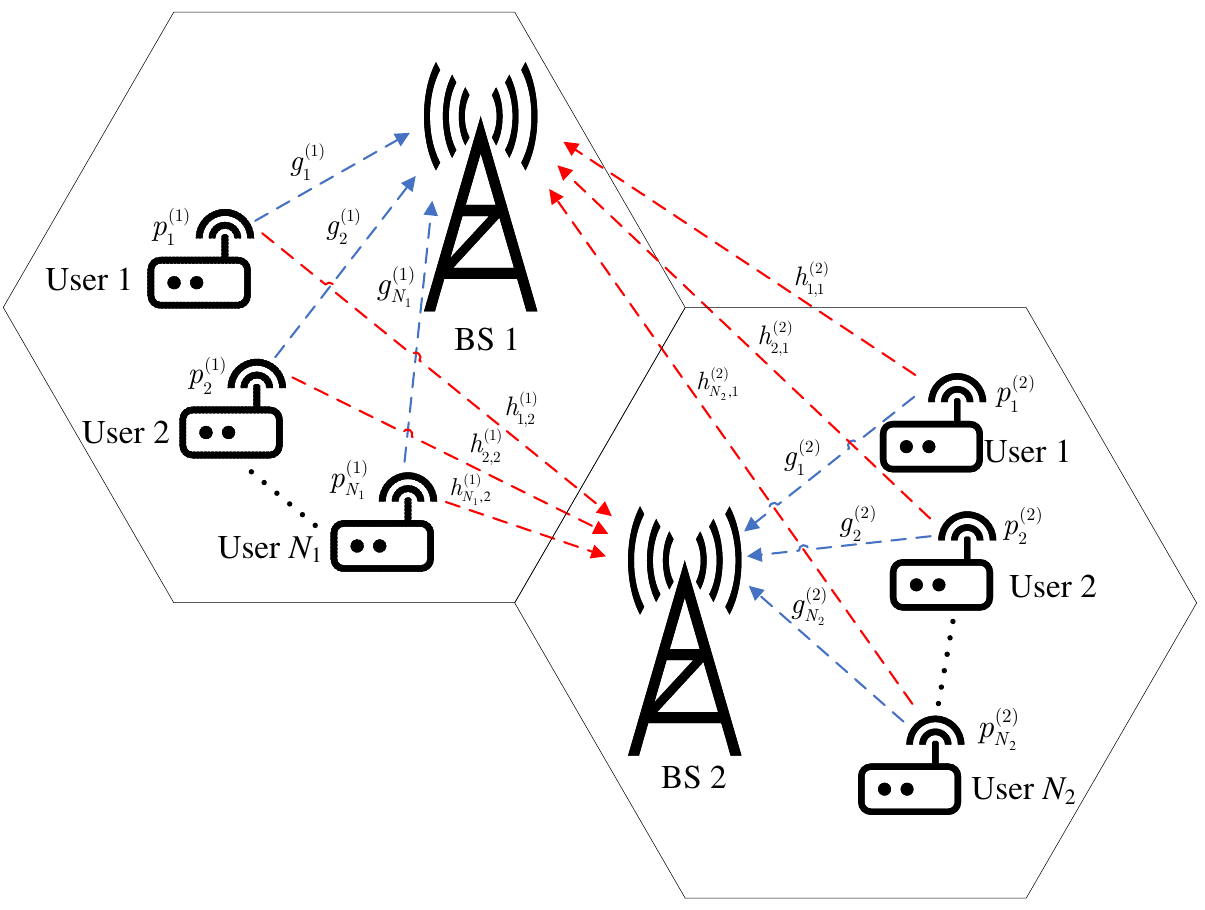}
		\caption{The system model of the dual-BS NOMA networks.}
		\label{fig:system_model_MultipleBSs}
	\end{figure}
	To demonstrate the effectiveness of the enhanced mask mechanism in ASOPA, let's consider a case study outlined in Table~\ref{tab:case_decoder_mutipleBS}. As depicted in Fig.~\ref{fig:system_model_MultipleBSs}, we consider a dual-BS NOMA network, where each BS contains three users, as \( N_b=3 \) for all \( b \in \mathbf{B} = \{1,2\} \). Consequently, it takes six iterations for ASOPA to establish the SIC ordering for all users. Utilizing the mask mechanism, ASOPA initially decodes the users in BS~1 during iterations \( t=1 \), 2, and 3, and then shifts to decoding users in BS~2 for iterations \( t=4 \), 5, and 6.	
	In the first iteration (\( t=1 \)), the algorithm calculates the probabilities \( \bm{\ell}^{(1)}_{1} \) for users in BS~1 from Equation~\eqref{equation_find_max_pro_user} and simultaneously masks the probabilities of users in BS~2 by setting \( \bm{\ell}^{(2)}_{1} = 0 \). Given that user 1 in BS~1 has the highest probability of 0.3997, it is selected as the first decoded user, \( \pi^{(1)}_{1} = 1 \). In the next two iterations, users in BS~2 remain masked, indicated by \( \bm{\ell}^{(2)}_{2} = 0 \) and \( \bm{\ell}^{(2)}_{3} = 0 \).	
	Conversely, when decoding the SIC ordering for users in BS~2 during iterations \( t=4 \), 5, and 6, the users in BS~1 are masked with \( \bm{\ell}^{(1)}_{t} = 0 \). This mechanism enhances the efficiency and accuracy of ASOPA in multi-BS NOMA networks by systematically focusing on one BS at a time, thereby streamlining the decoding process.}

	\section{Numerical Results} \label{sec:NumResults}
	In this section, we evaluate the proposed ASOPA algorithm through simulations in uplink NOMA networks. In these simulations, users are uniformly deployed within a 100-meter radius circle, with a BS at the center. The average channel gain, \(\overline{g}_n\), adheres to the free-space path loss model, following \(\overline{g}_n=A_d\left(\frac{3\cdot10^8}{4\pi f_c b_n}\right)^{b_e}\) \cite{huang2019deep}, where \(A_d=4.11\) represents the antenna gain, \(f_c=915\) MHz is the carrier frequency, \(b_n\) is the distance between each user and the BS, and \(b_e=2.8\) is the path loss exponent. Each user \(n\)’s wireless channel gain, \(g_n\), is modeled as a Rayleigh fading channel, expressed as \(g_n=\overline{g}_n|\alpha_n|^2\), with \(|\alpha_n|^2\) being an independent random channel fading factor following an exponential distribution with unit mean. The system parameters include a bandwidth \(B\) of 1 MHz and a noise power spectral density of \(-174\) dBm/Hz. Each user's maximum power is capped at \(P_n^{max}=1\) Watt, and the user weight \(w_n\) is chosen from the set \(\{1,2,4,8,16,32\}\).
	
	{
		For the neural network training, the samples \(\mathbf{X}\) arrive in each time slots and are stored in a replay memory of size 1280. The number of users \(N\) in each sample varies uniformly between 5 and 10. The batch size for once policy update is set to \(\left | \bm{\tau} \right | = 64\), and each training epoch consists of $M=20$ times policy updates. After each training epoch, the baseline network updates its parameters \(\theta'\). The learning rate for the Adam optimizer is set at 1e-4, and the embedding dimension for users in the actor network is \(d_e=128\). 
		The simulations are carried out on a desktop with an Intel Core i7-10700 2.9 GHz CPU, 32 GB memory, and an NVIDIA GeForce RTX 3060 Ti GPU, ensuring robust computational performance.}	
	{ The source code for ASOPA is accessible at https://github.com/Jil-Menzerna/ASOPA.}
	
	\subsection{Convergence Performance}
	\begin{figure*}
		\centering
		\includegraphics[width=1.03\linewidth]{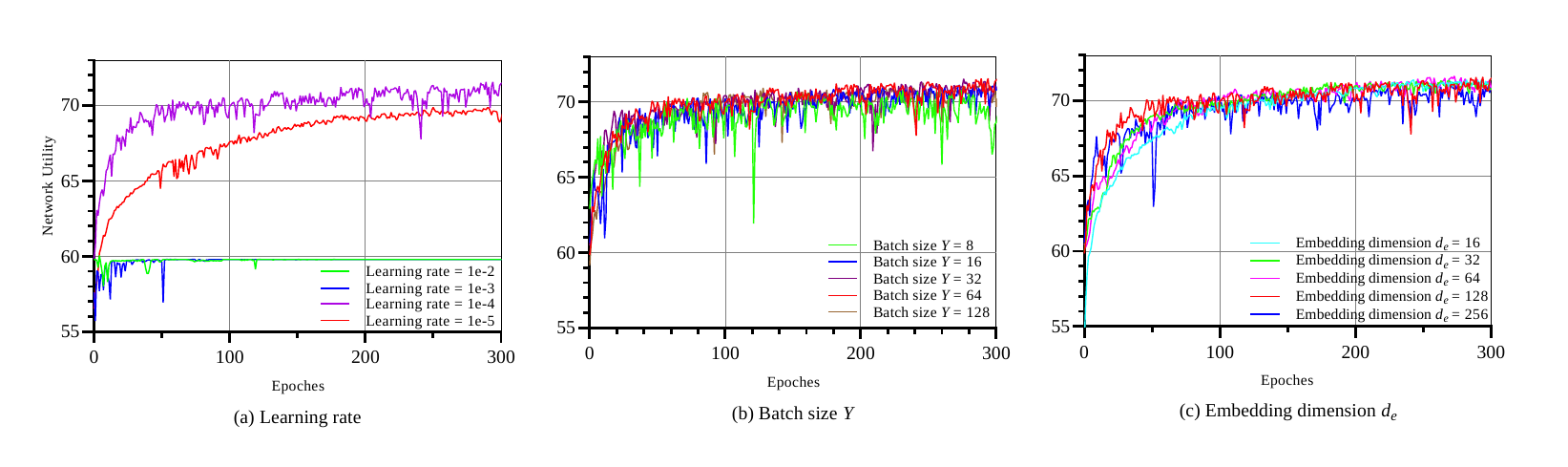}
		\caption{Convergence performance of ASOPA under different algorithm parameters when $N=10$.}
		\label{fig:convergence_performance}
	\end{figure*}
	
	In Fig.~\ref{fig:convergence_performance}, we evaluate the effect of different parameters on the convergence performance of ASOPA, including different learning rates, batch sizes, and embedding dimensions.
	
	Fig.~\ref{fig:convergence_performance}(a) shows the effect of different learning rates. We can see that a significant learning rate (1e-2 or 1e-3) causes the algorithm to converge to a local optimum, but a small learning rate (1e-5) results in slow convergence. Hence, the learning rate is set as 1e-4.
	
	Fig.~\ref{fig:convergence_performance}(b) shows the effect of different batch sizes. A small batch size (8 or 16) leads to high variance in the network utility. 
	The larger the batch size, the more memory space the algorithm consumes. Also, a large batch size may reduce the randomness of gradient descent and lead to the local minimum value. Hence, the batch size is set to 64.
	
	Fig.~\ref{fig:convergence_performance}(c) shows the effect of different embedding dimensions $d_e$. A small embedding dimension (16) cannot adequately characterize features and thus degrades the performance and convergence speed. A large embedding dimension (256) may overfit the training set, resulting in unstable performance. Hence, the embedding dimension is set as $d_e=128$.
	
	Overall, the simulation results in Fig.~\ref{fig:convergence_performance} show that the proposed ASOPA can converge under the set parameters.
	
	\subsection{Network Utility Performance}
	To evaluate the SIC ordering generated by ASOPA, we compare it with five baseline algorithms:
	\begin{enumerate}
		\item Exhaustive search \cite{hu2019joint}: This scheme calculates the network utility for all $N!$ SIC orderings and obtains the optimal network utility.
		{
		\item Tabu search \cite{QianL2024}: This scheme initiates a SIC ordering and swaps any two users' ordering to search. For each search iteration, it tries all possible swapping of two users for a SIC ordering and selects the best one for the next search iteration.
		To the best of our knowledge, the Tabu search algorithm presented in \cite{QianL2024} is the state-of-the-art algorithm for dynamic SIC ordering, albeit at the cost of high computational complexity.}
		\item Meta-scheduling \cite{qian2019optimal}:  This scheme sequentially adds and inserts each user into an order. It tries every possible insertion position for each insertion and greedily chooses the one with the greatest utility gain. 
		\item Weight descending \cite{ding2020unveiling2}: The static SIC ordering follows the descending order of users' weights.
		\item Channel descending \cite{duan2019resource}: The static SIC ordering follows the descending order of users' channel gains.
	\end{enumerate}

	After those baseline algorithms determine the SIC ordering, the optimal transmit power is determined by the power allocation method proposed in Section~\ref{sec_convex_opt}.
	
	\begin{figure}
		
		\centering
		\includegraphics[width=1\linewidth]{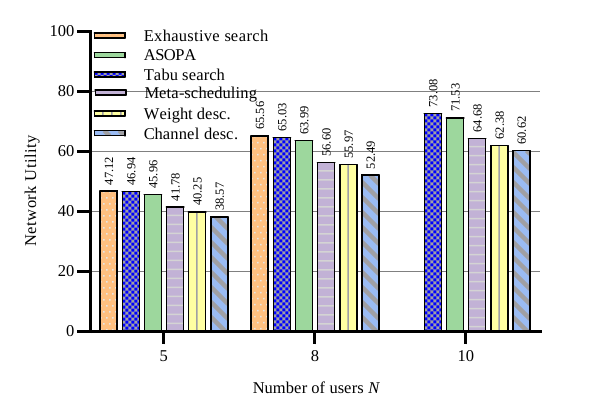}
		\caption{The network utility under different numbers of users.}
		\label{fig:performance_duibi_users_num}
	\end{figure}
	
	Fig.~\ref{fig:performance_duibi_users_num} presents the network utility achieved by different algorithms for varying numbers of users $N$. Exhaustive search achieves optimal performance with $N=5$ and $N=8$ but not $N=10$ due to the unacceptable running time for enumerating $10!$ possible SIC ordering.
	{
	When $N=5$ and $N=8$, through sufficient search iteration, Tabu search achieves 99.61\% and 99.19\% of the optimal performance obtained by exhaustive search. ASOPA achieves 97.54\% and 97.60\% of the optimal performance, which is close to the performance of Tabu search. 
	When $N=5$, $N=8$, and $N=10$, ASOPA is over 10\% higher in network utility than the other three baseline algorithms besides Tabu search, respectively.}
	
	\begin{figure}
		
		\subfigure[Convergence performance of ASOPA.]
		{
			\label{fig:convergence_large_scale}
			\includegraphics[width=0.95\linewidth]{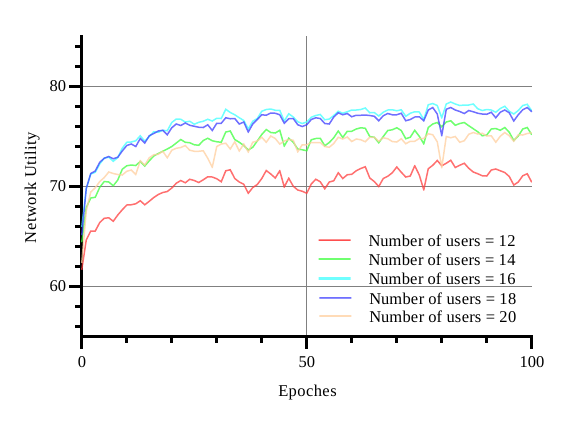}
		}
		\subfigure[The network utility of different algorithms.]
		{
			\label{fig:performance_large_scale}
			\includegraphics[width=0.93\linewidth]{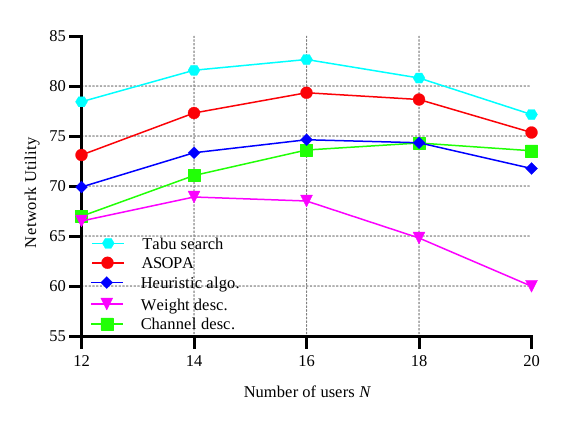}
		}
		\caption{The performance of ASOPA in large-scale scenarios when $N$ is between 10 and 20.}
		\label{fig:convergence_performance_large_scale}
	\end{figure}
	
	{
		Fig.~\ref{fig:convergence_performance_large_scale} provides further evaluation of ASOPA in large-scale scenarios, specifically where the number of users \(N\) varies from 10 to 20. In Fig.~\ref{fig:convergence_large_scale}, ASOPA demonstrates a consistent convergence rate of around 50 epochs, regardless of the specific values of \(N\). 
		{ 
		Meanwhile, Fig.~\ref{fig:performance_large_scale} illustrates that ASOPA consistently achieves average 95\% performance of Tabu search, and outperforms the other three baseline algorithms in these large-scale scenarios, aligning with the observations from Fig.~\ref{fig:performance_duibi_users_num}. }
		An interesting observation is that for \(N > 16\), the channel descending algorithm surpasses Meta-scheduling in terms of network utility. 
		This comparison further underscores that network utility tends to decline when the number of users in a NOMA system exceeds a certain threshold, particularly when \(N > 16\).
		{
		The decline of network utility can be attributed to the concave logarithm throughput function $\ln{R_n}$ in network utility in (3). Intuitively, as the number of users increases, the sum rate $\sum_{n=1}^N R_n$ tends to saturate, leading to a decrease in the sum of logarithms $\sum_{n=1}^N \ln R_n$ due to Jensen's inequality.}
		Overall, the results depicted in Fig.~\ref{fig:convergence_performance_large_scale} confirm ASOPA's effectiveness in handling large-scale scenarios and its superiority over all baseline algorithms.
	}

	\begin{figure}
		
		\subfigure[Boxplot of the normalized network utility for ASOPA and other baseline algorithms when $N=5$. The central line indicates the median, while the bottom and top edges of the box indicate the $25th$ and $75th$ percentiles, respectively.]
		{
			\label{fig:performance_boxline}
			\includegraphics[width=0.93\linewidth]{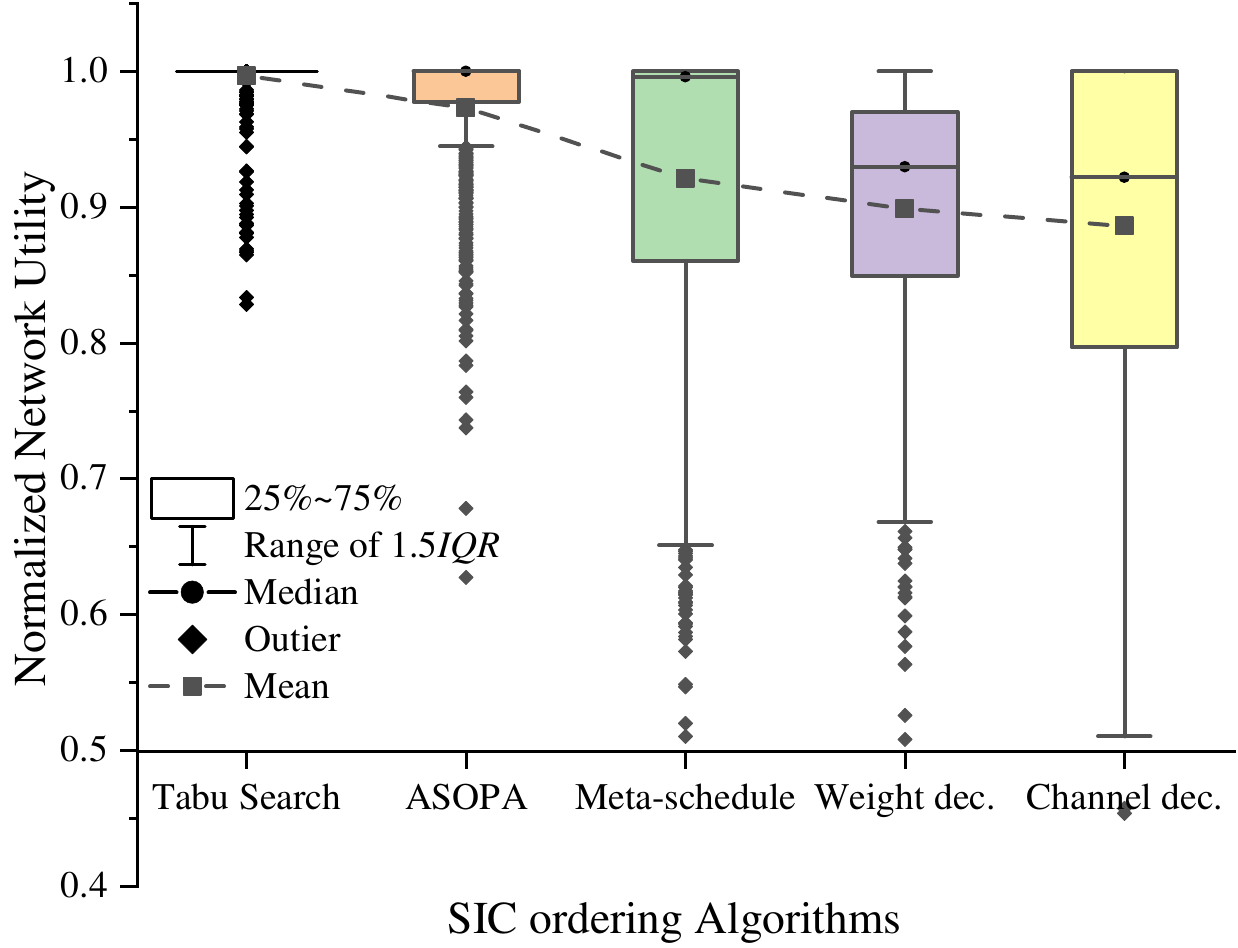}
		}
		\subfigure[Stacked plot of the top ten best SIC ordering hits when $N=5$.]
		{
			\label{fig:performance_topten}
			\includegraphics[width=1\linewidth]{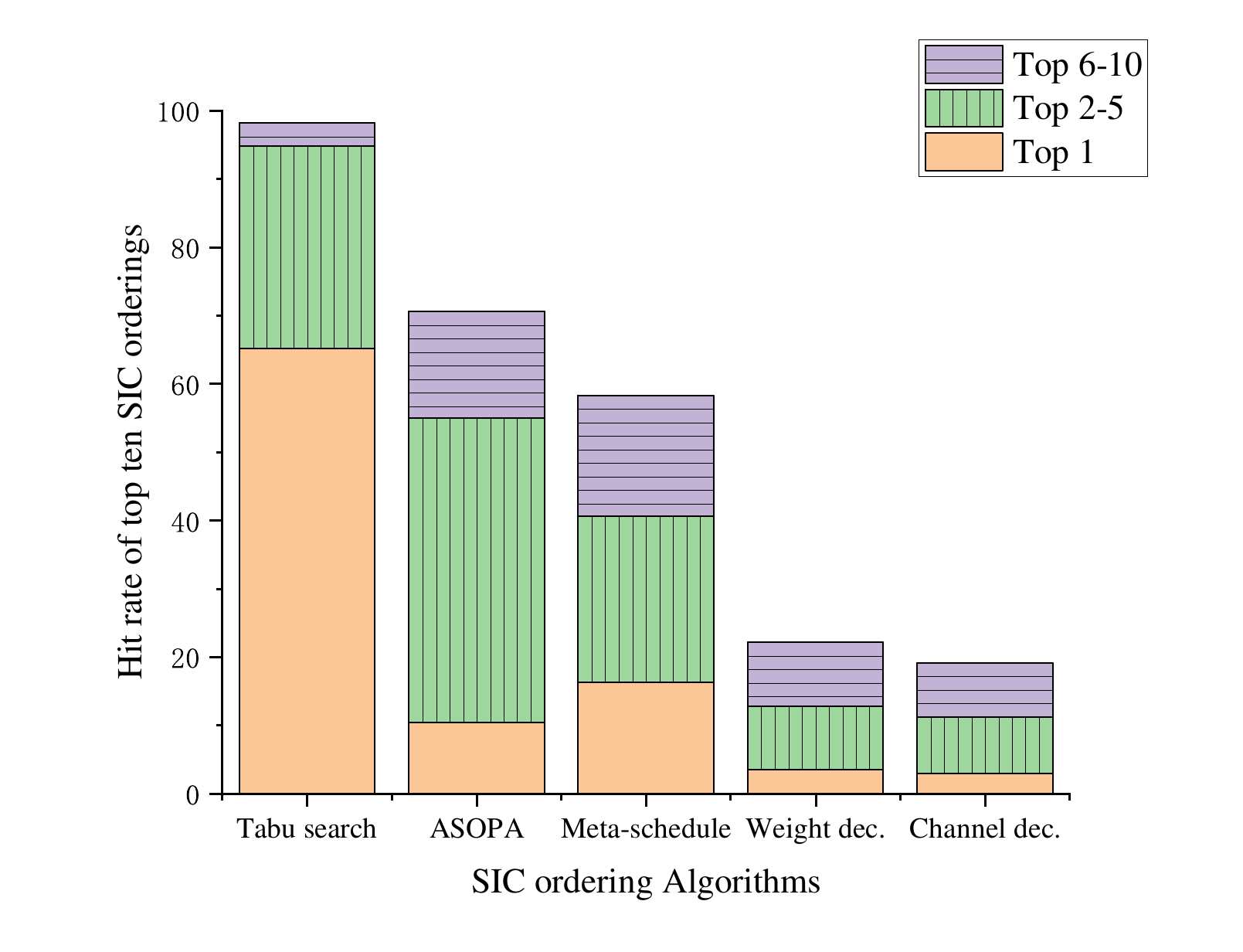}
		}
		\caption{The distribution of the network utility achieved by different algorithms.}
		\label{fig:performance_duibi_ES_norm}
	\end{figure}
	
	In Fig.~\ref{fig:performance_duibi_ES_norm}, we further compare the performance of ASOPA and baseline algorithms over 1000 independent samples when $N=5$. 
	Fig. \ref{fig:performance_boxline} displays the mean, median, confidence interval, and outliers of the normalized network utility for different algorithms. The normalized network utility is the ratio of the network utility achieved by an algorithm to the optimal network utility obtained by exhaustive search.
	{
	We observe that the medians of Tabu search and ASOPA are close to 1, and the confidence intervals of Tabu search and ASOPA are over 99\% and 97\%, respectively. Although some outliers affect the mean of ASOPA, it still outperforms the baseline algorithms besides Tabu search. 
	In Fig. \ref{fig:performance_topten}, we present the hit rate of the top 10 maximum network utilities for ASOPA and baseline algorithms. The hit rate is defined as the percentage of times that an algorithm generates an SIC ordering that appears in the top 10 maximum network utilities obtained by exhaustive search. We observe that ASOPA achieves hit rates of over 55\% and 70\% for the top 5 and top 10 maximum network utilities, respectively. }
	The results in Fig. \ref{fig:performance_duibi_ES_norm} further confirm that ASOPA can achieve near-optimal network utility performance.

	\begin{figure}
		
		\centering
		\includegraphics[width=1\linewidth]{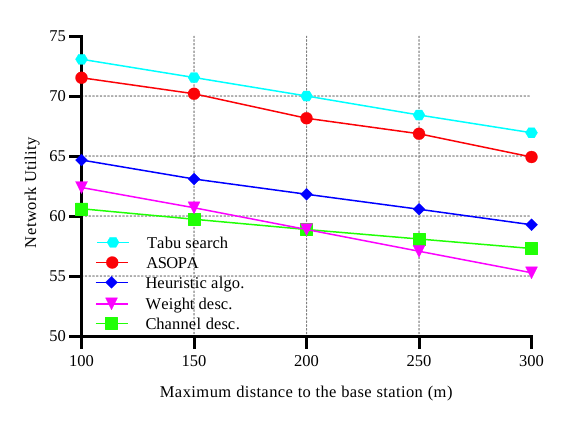}
		\caption{The network utility under different maximum distances to the BS when $N=10$.}
		\label{fig:performance_duibi_d_max}
	\end{figure}
	
	\begin{figure}
		
		\centering
		\includegraphics[width=1\linewidth]{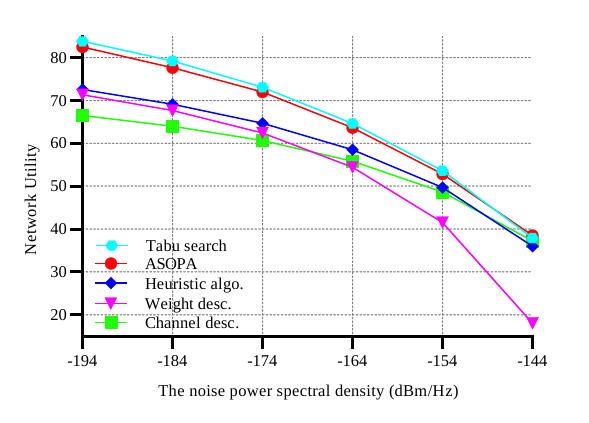}
		\caption{The network utility under different noise power spectral densities when $N=10$.}
		\label{fig:performance_duibi_noise}
	\end{figure}
	
	Fig.~\ref{fig:performance_duibi_d_max} shows the network utility under different maximum distances of users to the BS. The network utility achieved by all algorithms decreases slightly as the maximum distance increases. {Within the distance range $[100, 300]~ \text{m}$, ASOPA achieves performance close to that of Tabu search algorithm and outperforms the other three baseline algorithms by an average of 10\%. }
	
	Fig.~\ref{fig:performance_duibi_noise} demonstrates how the network utility varies with different levels of noise power spectral densities. As the noise power spectral density increases, the network utility of all algorithms decreases, and the difference between ASOPA and baseline algorithms diminishes. At a noise power spectral density of -144 dBm/Hz, ASOPA and the channel descending algorithm exhibit the slightest difference of 4.27\%. This result indicates that when the users' channel quality is inferior, the users' channel state significantly impacts the SIC ordering.

	\subsection{Execution Latency} \label{sec:latency}

	In order to meet the real-time requirement of NOMA networks, the execution time of the SIC ordering and power allocation algorithm need be much smaller than the slot duration, i.e., two seconds \cite{huang2019deep}. To evaluate the efficiency of ASOPA and baseline algorithms, we test the average execution time under different numbers of users, and the results are shown in Fig.~\ref{fig:speed} and Table.~\ref{Table:latency}. 
	\begin{figure}
		
		\centering
		\includegraphics[width=1\linewidth]{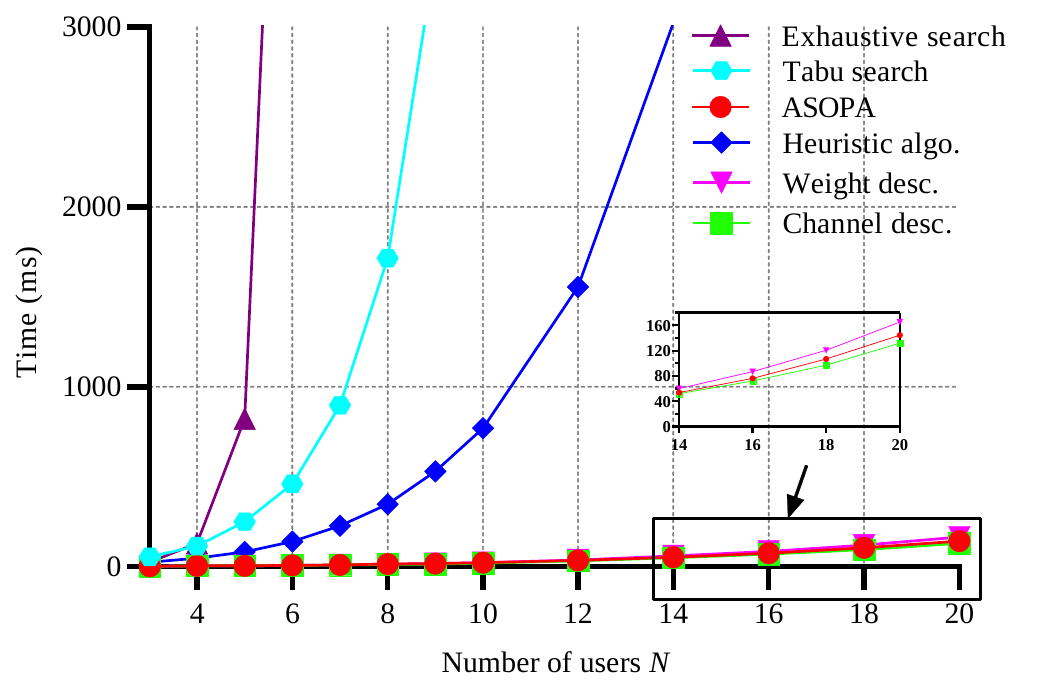}
		\caption{The average execution time under different numbers of users.}
		\label{fig:speed}
	\end{figure}
	
	\begin{table}
		\setlength{\abovecaptionskip}{0cm}
		\setlength{\belowcaptionskip}{0.1cm}
		\renewcommand{\arraystretch}{1.2}
		\caption{The average execution latency of different algorithms (ms)} 
		\centering 
		
		\label{Table:latency}
		\begin{tabular}{crrrrr} 
			\toprule 
			Algorithms & N=5 & N=8 & N=10&N=14&N=20\\
			\midrule
			Exhaustive search &823  &$741\,933$ & /& /& /\\ [2pt]
			Tabu search &252  &$1\,717$  &$5\,643$ &$45\,122$ & $492\,677$\\ [2pt]
			Meta-scheduling &84 &349 & 771 &$3\,025$ & $14\,737$  \\ [2pt]
			Weight des. &7  & 15 &25  & 60&165  \\ [2pt]
			Channel des. & 7 &14  & 23 &52 &132  \\ [2pt]
			\hline
			\specialrule{0em}{1pt}{1pt}
			ASOPA &8  & 16  &24  &53 & 152 \\
			\bottomrule 
		\end{tabular}
	\end{table}

	The execution time of ASOPA is close to that of the weight descending algorithm and the channel descending algorithm, i.e., 24 ms, 25 ms, and 23 ms for a ten-user NOMA network. The execution delay is scalable with the network size $N$ and is acceptable for field deployment. { ASOPA only takes 152 ms even for a twenty-user NOMA network.} { However, the execution latency of Meta-scheduling and Tabu saerch significantly increases with the network size $N$, consuming 771 ms and $5643$ ms for a ten-user NOMA network, respectively.} The execution time of exhaustive search is exponentially increasing with $N$. It takes 823 ms and $6630$ ms even for NOMA networks with five-user and six-user, respectively.
	{
	According to Fig.~\ref{fig:speed} and Table.~\ref{Table:latency}, Tabu search fails to cope with real-time execution when $N>7$, while Meta-scheduling fails at $N>10$. In contrast, ASOPA maintains the same latency as the static algorithm for all the number of users. In particular, is three orders of magnitude lower than Tabu search and two orders of magnitude lower than Meta-scheduling when \(N=20\).}
	
	ASOPA uses the actor network to generate the SIC ordering, whose time consumption is negligible. The primary time overhead of various algorithms comes from solving the power allocation problem by the interior-point method. {Exhaustive search solves the power allocation problem $N!$ times. Meta-scheduling solves the power allocation ${N(N+1)}/{2}$ times. Tabu search solves the power allocation ${IN(N+1)}/{2}$ times, where $I$ denotes the number of search iterations. ASOPA, the weight descending algorithm, and the channel descending algorithm solve the power allocation problem once.} Therefore, the proposed ASOPA executes efficiently like the static SIC ordering algorithms, while performing as well as the exhaustive search algorithm.
	
	{Regarding the training latency, ASOPA's policy update is conducted infrequently and in parallel with the inference process, as detailed in Algorithm~\ref{algorithm_train}. Extensive evaluations have shown that the duration of a single policy update is approximately one second when the number of users \(N\) is 10 and training batch size is 64. On average, the duration of the policy update process is less than 20 ms for each sample. Therefore, the policy update process of ASOPA can feasibly be executed online for NOMA networks, ensuring that the system remains up-to-date and responsive to changing network conditions.}

	{
		\subsection{Extension Scenario}		
		\begin{figure}
			
			\centering
			\includegraphics[width=1\linewidth]{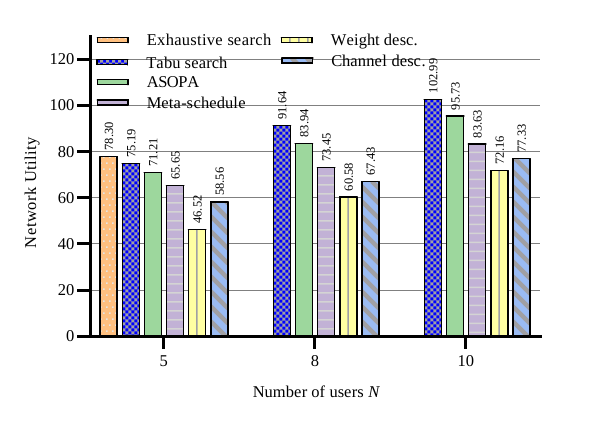}
			\caption{The network utility under different numbers of users for multiple-antenna NOMA networks.}
			\label{fig:MMSE}
		\end{figure}
		
		Fig.~\ref{fig:MMSE} presents the network utility achieved by different algorithms for varying numbers of users $N$ in multiple-antenna scenario. Specifically, we consider a NOMA network with two antennas at the BS. All algorithms use the minimum-mean-square-error (MMSE) \cite{ZhangJ2022} linear equalization to detect symbols.
		For $N=5$, the optimal network utility was calculated using exhaustive search. Remarkably, ASOPA achieves 90.94\% of the optimal network utility, with only a 5\% performance degradation compared to the Tabu search algorithm.
		Furthermore, ASOPA consistently outperforms the other three baseline algorithms across all settings, which agrees the observation in Fig.~\ref{fig:performance_duibi_users_num}.
		Due to the complexity of user state in multiple antenna scenarios, the performance of ASOPA can be further optimized in future work.
		
		\begin{figure}
			
			\centering
			\includegraphics[width=1\linewidth]{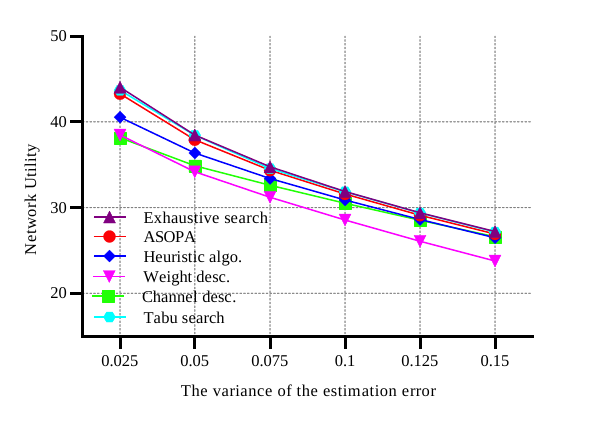}
			\caption{The network utility under different estimated errors when $N=5$.}
			\label{fig:imperfectCSI}
		\end{figure}

		
		Fig.~\ref{fig:imperfectCSI} shows how network utility varies with different variances of estimated error under imperfect CSI conditions with five users. { ASOPA consistently achieves over 98\% of the optimal performance obtained through exhaustive search and achieves 99\% performance of Tabu search.  Specifically, when \( \sigma _{{\epsilon_n}}^2 = 0.025\), ASOPA's network utility is 6.75\%, 12.49\%, 13.44\% higher than that of Meta-scheduling, channel descending and weight descending, respectively. }
		These results demonstrate ASOPA's robustness and effectiveness in scenarios with imperfect CSI, highlighting its ability to adapt to varying degrees of channel estimation errors.

		\begin{figure}
			\centering
			\includegraphics[width=1\linewidth]{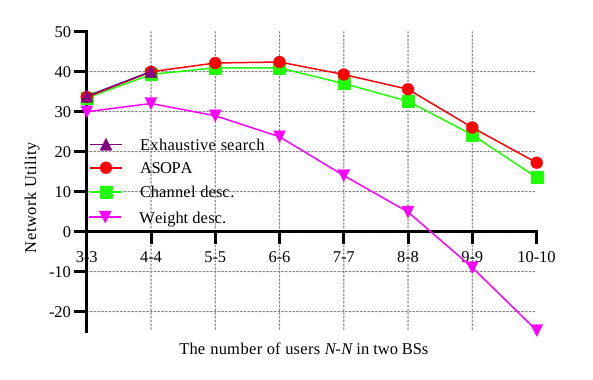}
			\caption{The network utility under different numbers of users in a dual-BS NOMA network.}
			\label{fig:multipleBS}
		\end{figure}
		Fig.~\ref{fig:multipleBS} presents the network utility achieved by various algorithms for different pairs of users $N_1$-$N_2$ in dual-BS NOMA networks. From the figure, it's evident that ASOPA performs comparably to the exhaustive search method and surpasses other benchmark algorithms. Notably, Tabu search in \cite{QianL2024} and Meta-scheduling proposed in \cite{qian2019optimal} is not applicable in this scenario and thus is not included in the comparison. Particularly, when each BS has three or four users, ASOPA achieves 99.45\% and 99.80\% of the optimal performance determined by exhaustive search, respectively. The network utility reaches its peak when each BS is serving six users. {This optimal network utilization can be attributed to the increasing inter-cell interference with the number of users and the logarithmic throughput function in network utility. As the number of users increases, although the sum throughput saturates, some users' weighted logarithmic throughput even becomes a small negative value (say~$-10$), which leads to a decline in network utility.} As the number of users per BS rises to ten, ASOPA's performance advantage becomes more pronounced, showing a 27.69\% higher network utility compared to the channel descending algorithm. These results highlight ASOPA's effectiveness in adapting to varying user densities in multi-BS NOMA networks.
	}

	\section{Conclusion} \label{sec:conclusion}
	This paper focuses on optimizing the sum-weighted logarithmic throughput in uplink NOMA-based wireless networks by jointly optimizing the SIC ordering and users' transmit powers. To tackle this problem, we propose the ASOPA framework, which innovatively combines DRL with optimization theory. Key to ASOPA's success is an attention-based actor network, trained via reinforcement learning, which effectively derives a near-optimal SIC ordering. Subsequently, this is complemented by the application of optimization techniques to allocate the optimal transmit power for users. Simulation results show that ASOPA can achieve near-optimal performance in a low execution latency. 
	A particularly noteworthy aspect of ASOPA is its extensibility; the framework is adept at solving a range of optimization challenges, particularly those that involve dynamic SIC orderings within the NOMA context.
	
	{
	Looking ahead, our aim is to evolve ASOPA for more complex scenarios, including developing a distributed framework for NOMA networks with multiple base stations, tackling the challenges of imperfect SIC decoding and integrating the QoS constraints in our framework. 
	}
	
	\appendices

	\section{Details of the Multi-head Attention Mechanism}\label{appendix_attention_layer_details}
	The input to the encoder is denoted as $\mathbf{X}$. It is first transformed into a $d_e$-dimensional space by a fully connected feed-forward ($\text{FF}_1$) layer. The output of the $\text{FF}_1$ layer is then fed into a multi-head attention (MHA) layer and a feed-forward ($\text{FF}_2$) layer in sequence. Therefore, the encoding process can be expressed as
	{
		\begin{eqnarray}
			\begin{split}
				&\hat{\mathbf{e}}_n ={\text {BN}}\bigg({\text {FF}}_1\big(\mathbf{x}_n\big)+{\text{MHA}}\Big({\text{FF}}_1\big(\mathbf{X}_1\big),...,{\text{FF}}_1\big(\mathbf{X}_n\big)\Big)\bigg)  \\
				&\mathbf{e}_n=\text{BN}\bigg(\hat{\mathbf{e}}_n+\text{FF}_2\big(\hat{\mathbf{e}}_n\big)\bigg).
			\end{split}
	\end{eqnarray}}

	For the MHA layer and the $\text{FF}_2$ layer, they have the residual connection (RC) and are followed by batch normalization (BN). The details of each layer are shown as follows.
	
	\subsection{Attention Mechanism}\label{appendix_attention_mechanism}
	We utilize the attention mechanism proposed in \cite{vaswani2017attention}.
	The attention mechanism computes a weighted sum of values, where the weight is determined by a compatibility function based on a query and a set of keys. The query, keys, and values are all embeddings. Specifically, we compute the query $\mathbf{q}_n$, the key $\mathbf{k}_n$, and the value $\mathbf{v}_n$ for each user $n$ by multiplying their respective embedding $\mathbf{e}_n$ with parameter matrices $\mathbf{W}^Q$, $\mathbf{W}^K$, and $\mathbf{W}^V$. These parameter matrices have sizes $d_k\times d_e$, $d_k\times d_e$, and $d_v\times d_e$, respectively, as	
	\begin{equation}
		\mathbf{q}_n=\mathbf{W}^{Q}\mathbf{e}_n,
		\mathbf{k}_n=\mathbf{W}^{K}\mathbf{e}_n,
		\mathbf{v}_n=\mathbf{W}^{V}\mathbf{e}_n.
	\end{equation}
	Then we compute the compatibility $u_{n,j}$ of user $n$'s query $\mathbf{q}_n$ with user $j$'s key $	\mathbf{k}_j$:
	\begin{equation}
		u_{n,j}=\frac{{\mathbf{q}_n}^T\mathbf{k}_j}{\sqrt{d_k}}.
	\end{equation}
	From $u_{n,j}$, we can compute the attention weight $a_{n,j}$ by a softmax function:
	\begin{equation}\label{equation_softmax}
		a_{n,j}=\frac{\exp\left(u_{n,j}\right)}{\sum_{j=1}^{N}\exp\left(u_{n,j}\right)}.
	\end{equation}
	Finally, we compute the sum of weighted keys to get the final message $\mathbf{e}_n'$:
	\begin{equation}
		\mathbf{e}_n'=\sum_{j=1}^{N} a_{n,j}\mathbf{v}_j.
	\end{equation}
	\subsection{Multi-head Attention}
	Multi-head attention uses $M$ groups of different parameters $\mathbf{W}_m^Q$, $\mathbf{W}_m^K$ and $\mathbf{W}_m^V$. We set $M=8$ and $d_k=d_v=\frac{d_e}{M}=16$, to get the messages, which are denoted as $\mathbf{e}_{n,m}',\forall m\in\{1,\dots,M\}$, and use $d_e\times d_v$ matrices $\mathbf{W}_m^A$ to change their size and then sum them up as the final message:
	\begin{equation}
		\text{MHA}_n\left(\mathbf{e}_1,\dots,\mathbf{e}_N\right)=\sum_{m=1}^{M}\mathbf{W}_m^A\mathbf{e}_{n,m}'.
	\end{equation}
	
	\subsection{Feed Forward Layer}
	There are two feed-forward layers in the encoder. The first $\text{FF}_1$  is just a fully connected layer with learnable parameters $\mathbf{W}_1$ and $\mathbf{b}_1$:
	\begin{equation}
		\text{FF}_1\left(\mathbf{x}_n\right)=\mathbf{W}_1\mathbf{x}_n+\mathbf{b}_1.
	\end{equation}
	And the  second feed-forward layer $\text{FF}_2$ consists of two fully connected layer and use a Relu activation after the first connected layer:
	\begin{equation}
		\text{FF}_2(\hat{\mathbf{e}}_n)=\mathbf{W}_{2,2}\text{Relu}\left(\mathbf{W}_{2,1}\hat{\mathbf{e}}_n+\mathbf{b}_{2,1}\right)+\mathbf{b}_{2,2},
	\end{equation}
	where $\hat{\mathbf{e}}_n$ is the input for $\text{FF}_2$, $\mathbf{W}_{2,1}$ and $\mathbf{b}_{2,1}$ are the parameter matrix and bias of the first fully connected layer, respectively, and $\mathbf{W}_{2,2}$ and $\mathbf{b}_{2,2}$ are the one of the second layer, respectively.
	
	\subsection{Batch Normalization }
	We use batch normalization shown in \cite{kool2018attention}:
	\begin{equation}
		\text{BN}\left(\mathbf{e}_n\right)=\mathbf{w}^{\text{bn}}\odot\overline{\text{BN}}\left(\mathbf{e}_n\right)+\mathbf{b}^{\text{bn}},
	\end{equation}
	where $\mathbf{w}^{\text{bn}}$ and $\mathbf{b}^{\text{bn}}$ are learnable $d_e$-dimensional affine parameters, $\odot$ denotes the element-wise product, and $\overline{\text{BN}}$ refers to batch normalization without affine transformation.

	{
		\section{Convex transformation and proof of \textbf{P1}}\label{appendix_proof_of_P2_is_convex}
	}	
	{{
			For user's transmit power $p_n>0,\forall n \in\mathcal{N}$, let $p_n=e^{y_n}, \forall n \in\mathcal{N}$. We introduce an auxiliary variable $\nu_n$ for user $n$ and add a constraint to guarantee that the weighted logarithmic throughput of user $n$ is not less than ${\nu_n}$. Then the power allocation sub-problem \textbf{P1} can be transformed into the following convex problem
			\begin{subequations}
				\begin{align}
					\textbf{P2}:&\max\limits_{\bm{\nu,y}}\ \sum\limits_{n=1}^{N}{\nu_n}\label{P2_objective}\\
					s.t.\ &e^{y_n}\le P_n^{max},\label{P2_constraint_power}\\
					&w_n\ln{\log_2\left(1+\frac{e^{y_n}g_n}{\sum\limits_{\xi\left(n'\right)>\xi\left(n\right),\forall n' \in\mathcal{N}}{e^{y_{n'}}g_{n'}}+N_0}\right)}\ge {\nu_n}, \label{P2_constrainsuzhu} \\
					&\forall n \in\mathcal{N}, \notag
				\end{align}
			\end{subequations}
			where $\bm{\nu} = [\nu_1,\nu_2,...,\nu_N]$ and $\bm{y} = [y_1,y_2,...,y_N]$.}
		It's easy to know that \eqref{P2_objective} and \eqref{P2_constraint_power} are convex. Next, we will show that \eqref{P2_constrainsuzhu} is convex.
		First we convert \eqref{P2_constrainsuzhu} as 
		\begin{equation}
			\frac{e^{y_n}g_n}{\sum\limits_{\xi\left(n'\right)>\xi\left(n\right),\forall n' \in\mathcal{N}}{e^{y_{n'}}g_{n'}}+N_0} \ge 2^{e^{\frac{\nu_n}{w_n}}}-1.
		\end{equation}
		Then we take the reciprocal of both sides and take the natural logarithm of both sides, so we can get
		\begin{equation}
			\ln\left(\frac{\sum\limits_{\xi\left(n'\right)>\xi\left(n\right),\forall n' \in\mathcal{N}}{e^{y_{n'}}g_{n'}}+N_0}{e^{y_n}g_ n}\right)+ \ln\left(2^{e^{\frac{\nu_n}{w_n}}}-1\right)\le 0.
		\end{equation}
		The first term in the left-hand-side (LHS) is a log-sum-exp function which is convex [52].
		The second-order derivative of the second term in the LHS is 
		\begin{equation}
			\label{second_term_LHS}
			\frac{\ln 2}{{w_n}^2}e^{\frac{\nu_n}{w_n}}2^{e^{\frac{\nu_n}{w_n}}}\left(2^{e^{\frac{\nu_n}{w_n}}}-\ln 2e^{\frac{\nu_n}{w_n}}-1\right),
		\end{equation}
		{
			whose value is non-negative. Since the first order derivative of the term inside brackets in (\ref{second_term_LHS}) is $\frac{\ln 2}{{{w}_{n}}}{{e}^{\frac{{{\nu }_{n}}}{{{w}_{n}}}}}\left( {{2}^{{{e}^{\frac{{{\nu }_{n}}}{{{w}_{n}}}}}}}-1 \right)$, which is positive due to ${{\nu }_{n}}>0,{{w}_{n}}>0$. Thus, the minimum of (\ref{second_term_LHS}) is $\ln{\frac{e}{2}}$ larger than zero, and the second term in the LHS is also convex. Therefore, \eqref{P2_constrainsuzhu} is convex. For \eqref{P2_objective}$\sim$\eqref{P2_constrainsuzhu} are convex, the problem \textbf{P2} is convex. The proof is completed. }

		{
			
			\section{Probabilistic problem Transformation}\label{Transformation_imperfectCSI}
			In the imperfect CSI scenario, the outage probability requirement turns the problem into an intractable non-convex probability mixed problem. Following \cite{NgDWK2012,ZhangH2020}, we transform this problem into a non-probability problem by approximations.
			
			Firstly, we transform the outage probabilistic requirement into another form of probabilistic constraint as follows.
			The maximum achievable data rate is rewritten as
			\begin{align}
				c_n &=  W\log_2(1+{\phi}_n) \notag\\
				&= W\log_2\left(1+\frac{c^S_n}{c^I_n}\right),
			\end{align}
			where $c^S_n = {{p_n}|\alpha_n|^2{{\overline g}_n}}$ and $c^{I}_n=\!\!\!\!\!\!\!{\sum\limits_{\xi \left( {n'} \right) > \xi \left( n \right),\forall n' \in {\cal N}} \!\!\!\!\!\!\!\!\!\!\!\!\!\!{{p_{n'}}|\alpha_n|^2{{\overline g}_{n'}}}  + {N_0}}$. The scheduled data rate can be rewritten as
			\begin{align}
				r_n &=  W\log_2(1+{\hat \phi}_n)\notag\\
				&= W\log_2\left(1+\frac{{b}^S_n}{{b}^I_n}\right),
			\end{align}
			and we have 
			\begin{align}
				{\hat \phi}_n = \frac{{b}^S_n}{{b}^I_n} = 2^{\frac{r_n}{W}}-1. \label{Eqn:SINR}
			\end{align}
			
			According to the above transformation and the total probability theorem, the outage probability constraints can be transformed as
			\begin{align}
				\Pr[c_n<r_n|{\hat \alpha}_n] =& \Pr\left[{\phi}_n<{\hat \phi}_n| {\hat \alpha}_n\right] \notag\\
				=&\Pr\left[\frac{c^S_n}{c^I_n}<2^{\frac{r_n}{W}}-1| {\hat \alpha}_n\right] \notag\\
				=& \Pr[E1]\cdot\Pr\left[c_n^S\le {b}^S_n | {\hat \alpha}_n\right] \notag\\
				&+  \Pr[E2]\cdot\Pr\left[c_n^S> {b}^S_n | {\hat \alpha}_n\right] \le \epsilon_{out}, \label{ProbConstraint}
			\end{align}
			where $\Pr[E_1] = \Pr\left[\frac{c^S_n}{c^I_n}<2^{\frac{r_n}{W}}-1|c_n^S\le {b}^S_n ,{\hat \alpha}_n\right]$ and $\Pr[E_2] = \Pr\left[\frac{c^S_n}{c^I_n}<2^{\frac{r_n}{W}}-1|c_n^S> {b}^S_n ,{\hat \alpha}_n\right]$. Then, we have the following theorem.
			
			\begin{theorem}
				\label{Theo:ProbConstraintApprox}
				Following \cite{ZhangH2020}, the outage probability constraint (\ref{ProbConstraint}) can be approximated as 
				\begin{align}
					\Pr\left[c_n^I \ge {b}^I_n | {\hat \alpha}_n\right] \le \epsilon_{out}/2, \label{Eqn:Approx1}
				\end{align}
				and 
				\begin{align}
					\Pr\left[c_n^S \le {b}^S_n | {\hat \alpha}_n\right] = \epsilon_{out}/2. \label{Eqn:Approx2}
				\end{align}
			\end{theorem}
			\begin{proof}
				According to (23), we have
				\begin{align}
					&\Pr\left[c_n^I \ge {b}^I_n | {\hat \alpha}_n\right]  \notag\\
					&\ \ \ \ \ \ \ \ \ \ = \Pr\left[c_n^I \ge {b}^S_n/(2^{\frac{r_n}{W}}-1)| {\hat \alpha}_n \right]\notag\\
					&\ \ \ \ \ \ \ \ \ \ =  \Pr\left[\frac{{b}^S_n}{c_n^I} \le 2^{\frac{r_n}{W}}-1 | {\hat \alpha}_n \right] \le \epsilon_{out}/2, 
				\end{align}
				and when $c_n^S> {b}^S_n$, we can always have 
				\begin{align}
					\Pr\left[E_2\right]=\Pr\left[\frac{c^S_n}{c^I_n}<2^{\frac{r_n}{W}}-1|{\hat \alpha}_n\right] \le \epsilon_{out}/2. \label{PrE2}
				\end{align}
				According to (24), we have 
				\begin{align}
					\Pr\left[ c_n^S > {b}^S_n | {\hat \alpha}_n \right] = 1-\epsilon_{out}/2. \label{Prc_g_a}
				\end{align}
				Based on (\ref{PrE2}) and (\ref{Prc_g_a}), the probabilistic constraint (\ref{ProbConstraint}) satisfies the following approximations
				\begin{align}
					&\Pr[c_n<r_n|{\hat \alpha}_n] \notag\\
					=& \Pr[E1]\cdot\Pr\left[c_n^S\le {b}^S_n | {\hat \alpha}_n\right] +  \Pr[E2]\cdot\Pr\left[c_n^S> {b}^S_n | {\hat \alpha}_n\right] \notag\\
					\le& \epsilon_{out}/2 + (\epsilon_{out}/2)(1-\epsilon_{out}/2) = \epsilon_{out} - \epsilon_{out}^2/4.
				\end{align}
				For $\epsilon_{out} \ll 1$, we have $\epsilon_{out} - \epsilon_{out}^2/4 \approx \epsilon_{out}$. Therefore,  the probabilistic constraint (\ref{ProbConstraint}) can be approximated as (\ref{Eqn:Approx1}) and (\ref{Eqn:Approx2}). This completes the proof.
			\end{proof}
			
			Secondly, based on the transformed probabilistic constraints (\ref{Eqn:Approx1}) and (\ref{Eqn:Approx2}) of Theorem~\ref{Theo:ProbConstraintApprox}, the probabilistic mixed problem can be further transformed to a non-probabilistic problem as follows. 
			
			According to the Markov inequality, the LHS of (\ref{Eqn:Approx1}) can derive as follows \cite{NgDWK2012}
			\begin{align}
				&\Pr\left[c_n^I \ge {b}^I_n | {\hat \alpha}_n\right] \notag\\
				&\ \ \ \ \ \ = \Pr\left[{\sum\limits_{\xi \left( {n'} \right) > \xi \left( n \right),\forall n' \in {\cal N}} \!\!\!\!\!\!\!\!\!\!\!\!\!\!{{p_{n'}}|\alpha_n|^2{{\overline g}_{n'}}}  + {N_0}}\ge {b}^I_n | {\hat \alpha}_n \right] \notag\\
				&\ \ \ \ \ \ \le \frac{E\left[{\sum\limits_{\xi \left( {n'} \right) > \xi \left( n \right),\forall n' \in {\cal N}} {{p_{n'}}|\alpha_n|^2{{\overline g}_{n'}}}}\right]}{{b}^I_n-N_0}\notag\\
				&\ \ \ \ \ \ = \frac{{\sum\limits_{\xi \left( {n'} \right) > \xi \left( n \right),\forall n' \in {\cal N}} {{p_{n'}}|\alpha_n|^2{{\overline g}_{n'}}}}}{{b}^I_n-N_0} = \epsilon_{out}/2, \label{Eqn:Markov}
			\end{align}
			where the right side of the Markov inequality is set to $\epsilon_{out}/2$ according to (\ref{Eqn:Approx1}).
			
			Since $\left| {\alpha}_n^2 \right|$ is a non-central chi-squared distributed random variable with two degrees of freedom, the LHS of (\ref{Eqn:Approx2}) can be rewritten as
			\begin{align}
				&\Pr\left[c_n^S \le {b}^S_n | {\hat \alpha}_n\right] \notag\\
				&\ \ \ \ \ = \Pr\left[{p_n}|\alpha_n|^2{{\overline g}_n}\le {b}^S_n | {\hat \alpha}_n \right] \notag\\
				&\ \ \ \ \ =\Pr\left[ {\left|\alpha_n\right|}^2 \le \frac{{b}^S_n}{p_n{\overline g}_n} | {\hat \alpha}_n \right] \notag\\
				&\ \ \ \ \ = F_{{\left|\alpha_n\right|}^2}\left(\frac{{b}^S_n}{p_n{\overline g}_n}\right) \notag \\
				& \ \ \ \ \ = 1- Q_1\left(\sqrt{\frac{2\left|{\hat \alpha}_n\right|^2}{\sigma_{\epsilon}^2}},\sqrt{\frac{2}{\sigma_\epsilon}\frac{{b}^S_n}{p_n{\overline g}_n}}\right) \label{Eqn:chisquare}
			\end{align}
			where $F(\cdot)$ denotes a cumulative distribution function (cdf) of a non-central chi-square random variable with non-centrality parameter ${2\left|{\hat \alpha}_n\right|^2}/{\sigma_{\epsilon}^2}$, and $Q_1(\cdot)$ is the first-order Marcum Q-function. Based on (\ref{Eqn:Approx2}), (\ref{Eqn:chisquare}) is equal to $\epsilon_{out}/2$, and ${b}^S_n$ can be expressed as
			\begin{align}
				{b}^S_n = F^{-1}_{{\left|\alpha_n\right|}^2}(\epsilon/2)\cdot p_n {\overline g}_n, \label{Eqn:inverseChiSquare}
			\end{align}
			where $F^{-1}(\cdot)$ is the inverse non-central chi-square cdf of $F(\cdot)$.
			Based on (\ref{Eqn:SINR}), (\ref{Eqn:inverseChiSquare}) and $|\alpha_n|^2 = |{\hat \alpha}_n|^2 + \sigma_{\epsilon}^2$ , (\ref{Eqn:Markov}) can be further transformed into
			\begin{align}
				&\frac{{\sum\limits_{\xi \left( {n'} \right) > \xi \left( n \right),\forall n' \in {\cal N}} {{p_{n'}}|\alpha_n|^2{{\overline g}_{n'}}}}}{b_n^{S}/(2^{\frac{r_n}{W}}-1)-N_0} \notag\\
				&\ \ \ \ \ = \frac{{\sum\limits_{\xi \left( {n'} \right) > \xi \left( n \right),\forall n' \in {\cal N}} {{p_{n'}}\left(|{\hat \alpha}_n|^2+\sigma^2_{\epsilon}\right){{\overline g}_{n'}}}}}{\frac{F^{-1}_{{|\alpha_n|}^2}\left(\epsilon_{out}/2\right)\cdot p_n{\overline g}_{n}}{2^{\frac{r_n}{W}}-1}-N_0} = \frac{\epsilon_{out}}{2}.
			\end{align}
			Therefore, considering the outage probability constraint, the approximated signal-to-interference-plus-noise ratio (SINR) ${\tilde \phi}_n$ for the $n$-th user can be derived as
			\begin{align}
				{\tilde \phi}_n = \frac{\epsilon_{out}{F^{-1}_{{|\alpha_n|}^2}\left(\epsilon_{out}/2\right)\cdot p_n{\overline g}_{n}}}{\epsilon_{out}N_0+2{\sum\limits_{\xi \left( {n'} \right) > \xi \left( n \right),\forall n' \in {\cal N}} {{p_{n'}}\left(|{\hat \alpha}_n|^2+\sigma^2_{\epsilon}\right){{\overline g}_{n'}}}}},
			\end{align}
			and the corresponding data rate can be written as
			\begin{align}
				\tilde{r}_n = W\log2(1+{\tilde \phi}_n).
			\end{align}
			
			Finally, the weighted proportional fairness function with outage probability is transformed into the following non-probability optimization problem
			\begin{subequations}
				\label{TransformedProblem}
				\begin{align}
					\max\limits_{\mathbf{\bm{\pi}},\textbf{p}}\ &\sum\limits_{n=1}^{N}{w_n\ln(1-\epsilon_{out}) {\tilde r}_n} \notag\\
					s.t.\ &0<p_n\le P_n^{max},\forall n \in\mathcal{N},\notag\\
					&\bm{\pi}\in\Pi. \notag
				\end{align}
			\end{subequations}

		}
		
		{
			\section{Alternative Algorithm for Multiple-antenna with MMSE equalization matrices}
			Under MMSE methods, the equalization matrices involving transmit power variable $\bf p$ turn the power allocation subproblem into an intractable non-convex problem. To tackle this non-convex problem, we utilize the alternative algorithm to further transform it into the following subproblem: the calculation of ${\bf V}$ under given $\bf p$ and the optimization of $\bf p$ under given ${\bf V}$.
			The specific process of the alternative algorithm is as follows.
			
			Firstly, we initiate the transmit power ${\bf p}$.
			
			Secondly, according to the definition, the equalization matrices ${\bf V}$ under the MMSE method can be easily calculated by the given $\bf p$ as
			\begin{equation}
				{\bf V} = {\bf{P}}{{\bf{H}}^H}{{({{\bf{H}}}{\bf{P}}{\bf{H}}^H + \sigma {\bf{I}})}^{ - 1}}
			\end{equation}
			
			Thirdly, obtained ${\bf V}$, the power allocation subproblem can be formulated as
			\begin{subequations}
				\label{Eqn:MMSEproblem}
				\begin{align}
					\textbf{P1}:R(\bm{\pi})=\max\limits_{\textbf{p}}\ &\sum\limits_{n=1}^{N}{w_n\ln{R_n}}\label{P1_objective}\\
					s.t.\ &0< p_n\le P_n^{max},\forall n \in\mathcal{N}\label{P1_constraint_power}.
				\end{align}
			\end{subequations}
			where $R_n$ is
			\begin{equation}
				{R_n}\!= \!{\log _2}\!\left(\!\! {1 \!+ \!\frac{{{{\left| {{{\bf{v}}_n}{{\bf{h}}_n}} \right|}^2}{p_n}}}{{\sum\limits_{\xi\left(n'\right)>\xi\left(n\right),\forall n' \in\mathcal{N}}\!\!\!\!\!\!\!\!\!\! {{{\left| {{{\bf{v}}_n}{{\bf{h}}_{n'}}} \right|}^2}{p_{n'}} + } {{\left| {{{\bf{v}}_{n}}} \right|}^2}{\sigma ^2}}}}\!\! \right)\!.
			\end{equation}
			${\bf v}_n$ denotes the $n$-row of the obtained ${\bf V}$. Since ${\bf v}_n$ is a constant under a given ${\bf V}$, (\ref{Eqn:MMSEproblem}) can be transformed into a convex problem as Appendix.~\ref{appendix_proof_of_P2_is_convex} and then solved by CVX solver.
			
			The alternative algorithm repeat the second and third steps above until the gap between the previous iteration's ${\bf p}$ and the current iteration's ${\bf p}$ is less than the threshold value.

		}

		{
			\section{Multiple-BS Scenario}\label{appendix_multipleBSs}
			We consider a uplink NOMA network consisting of a set of BSs $\mathbf{B}$ and each BS $b$ is associated with $N_b$ users. A BS simultaneously receives signal from its associated users and the other users, and iteratively decodes signal via a SIC ordering. In the SIC process, the remaining undecoded signal and the other users’ signal are treated as interference. Therefore, the SINR between user $n$ and BS $b$ can be expressed as
			\begin{align}
				{\phi^{(b)} _{n}} = \frac{{{p}^{(b)}_{n}}{g^{(b)}_{n}}}{{\sum\limits_{\!\!\!\!\!\!\xi \left( {n'} \right) > \xi \left( n \right),\forall n' \in {\cal N}_{b}} {{p^{(b)}_{n'}}{g^{(b)}_{n'}}}  + \sum\limits_{b' \in \mathbf{B}\setminus b,\forall m \in {\cal N}_b} {{p^{(b')}_{m}}{h^{(b')}_{m,b}}}  + {N_0}}}. \notag
			\end{align} 
			Then, we have the data rate between user $n$ associated with BS $b$,
			\begin{align}
				{R^{(b)}_{n}} = {\log _2}(1+\phi^{(b)}_{n}).
			\end{align}
			Therefore, the joint SIC ordering and power allocation optimization problem for multiple-BS NOMA can be expressed as:
			\begin{subequations}
				\begin{align}
					&\mathop {\max }\limits_{{\bf{\pi }},{\bf{p}}} \sum\limits_{b \in {\bf{B}}} {\sum\limits_{n = 1}^{N_b} {{w^{(b)}_{n}}\ln {R^{(b)}_{n}}} } \notag \\
					s.t.\ \ &0 < {p^{(b)}_{n}} \le P_{n,max}^{{(b)}},\forall n \in N_b,\forall b \in {\bf{B}} \\
					&\pi  \in {\Pi _ {\mathbf{B}} } 
				\end{align}
			\end{subequations}%
			where ${w^{(b)}_{n}}$ is the weight of user $n$ associated with BS $b$ and ${\bf{\pi }} = \left[\{\pi^{(b)}_{n}\}_{n\in\mathcal{N}_b}| b \in \bf{B} \right]$ indicates the SIC order. Here ${\pi^{(b)}_{i}} = n $  means that the $i$-th decoded user in BS $b$ is user $n$ and  ${\Pi _{\mathbf{B}}}$ is the permutation set of all possible SIC orderings.	
		}

	\end{document}